\PassOptionsToPackage{dvipsnames}{xcolor}
\documentclass[acmsmall,screen]{acmart}

\acmJournal{PACMPL}
\AtBeginDocument{%
  }

\setcopyright{acmlicensed}
\copyrightyear{2018}
\acmYear{2018}
\acmDOI{XXXXXXX.XXXXXXX}

\usepackage{booktabs}
\usepackage{placeins}
\usepackage{float}
\usepackage{listings}
\usepackage{xcolor}
\usepackage{microtype}
\usepackage{enumitem}
\usepackage{tikz}
\usetikzlibrary{shapes.geometric,arrows.meta,positioning,calc,fit}
\usepackage{algpseudocode}
\usepackage{xspace}
\newcommand{\eg}{\textit{e.g.,}}
\lstdefinelanguage{Rust}{
  keywords={fn,pub,unsafe,extern,let,mut,use,struct,impl,for,if,else,
             return,match,Some,None,true,false,const,static,type,mod,
             crate,self,super,as,in,where,trait,enum,Box,Vec,Option,
             Result,bool,u8,u16,u32,u64,i8,i16,i32,i64,usize,isize,
             str,String,CStr,CString},
  sensitive=true,
  morecomment=[l]{//},
  morecomment=[s]{/*}{*/},
  morestring=[b]",
}
\lstdefinelanguage{C}{
  keywords={int,char,void,struct,typedef,static,const,if,else,return,
             NULL,size_t,unsigned,long},
  sensitive=true,
  morecomment=[l]{//},
  morecomment=[s]{/*}{*/},
  morestring=[b]",
}
\usepackage[utf8]{inputenc} % allow utf-8 input
\usepackage{amsfonts}       % blackboard math symbols
\usepackage{multirow}       % for multirow cells in tables
\usepackage{nicefrac}       % compact symbols for 1/2, etc.

\usepackage{makecell}
\usepackage{enumitem}
\usepackage{amsmath} 
\newtheorem{definition}{Definition}
\newtheorem{invariant}{Invariant}
\newtheorem{claim}{Claim}
\usepackage{bm}
\usepackage[ruled,vlined,linesnumbered,noresetcount]{algorithm2e}
\usepackage{subcaption}
\begin{document}

\title{ENCRUST: Encapsulated Substitution and Agentic Refinement on a Live Scaffold for Safe C-to-Rust Translation}

\author{HoHyun Sim}
\email{tlaghgus0425@korea.ac.kr}
\authornotemark[1]
\affiliation{%
  \institution{Korea University}
  \country{South Korea}
}

\author{Hyeonjoong Cho}
\authornotemark[2]
\authornote{Corresponding author: Hyeonjoong Cho}
\email{raycho@korea.ac.kr}
\affiliation{%
  \institution{Korea University}
  \country{South Korea}
}

\author{Ali Shokri}
\authornotemark[3]
\email{ashokri@Central.UH.EDU}
\affiliation{%
  \institution{University of Houston}
  \country{United States}
}

\author{Zhoulai Fu}
\authornotemark[4]
\email{zhoulai.fu@sunykorea.ac.kr}
\affiliation{%
  \institution{State University of New York Korea}
  \country{South Korea}
}
\author{Binoy Ravindran}
\authornotemark[5]
\email{binoy@vt.edu}
\affiliation{%
  \institution{Virginia Tech}
  \country{United States}
}


\begin{abstract}

We present Encapsulated Substitution and Agentic Refinement on a Live Scaffold for Safe C-to-Rust Translation, a two-phase pipeline for translating real-world C projects to safe Rust. Existing automated approaches either produce wholly unsafe output that offers no memory-safety guarantees, or translate functions in isolation and verify them independently, failing to detect cross-unit type mismatches or handle unsafe constructs that inherently require whole-program reasoning. Furthermore, function-level LLM pipelines require coordinated caller updates whenever a type signature changes, while project-scale systems fail to produce compilable output under real-world dependency complexity. Encrust addresses these limitations by decoupling boundary adaptation from function logic via an Application Binary Interface (ABI)-preserving wrapper pattern and validating every intermediate state against the fully integrated codebase.
Phase 1 (Encapsulated Substitution) translates each function using an ABI-preserving wrapper pattern that splits it into two components: a caller-transparent shim retaining the original raw-pointer signature, and a safe inner function targeted by the LLM with a clean, scope-limited prompt. This design enables independent per-function type-signature changes with automatic rollback on failure, without requiring coordinated caller updates. A subsequent deterministic, type-directed wrapper elimination pass then removes the wrappers once translation succeeds.
Phase 2 (Agentic Refinement) resolves unsafe constructs that exceed per-function scope, including static mut globals, skipped wrapper pairs, and failed translations, through an LLM agent operating on the whole codebase under a baseline-aware verification gate.
We evaluate Encrust on 7 GNU Coreutils programs and 8 libraries from the Laertes benchmark (197,706 LoC, 2,366 functions), demonstrating substantial unsafe-construct reduction across all 15 programs while maintaining full test-vector correctness across all benchmarks.
\end{abstract}



\received{20 February 2007}
\received[revised]{12 March 2009}
\received[accepted]{5 June 2009}

\maketitle

\section{Introduction}

C remains one of the most widely deployed systems languages, forming the foundation of operating systems, embedded firmware, cryptographic
libraries, and critical infrastructure. Yet C provides no bounds checking, no ownership model, and no lifetime tracking; these missing guarantees are the root cause of a
substantial fraction of real-world security vulnerabilities: buffer
overflows, use-after-free errors, and data races continue to dominate
CVE databases despite decades of tooling investment~\cite{szekeres2013sok, cling, intel}.
Rust offers a compelling alternative, providing memory and
thread safety through a borrow-checked ownership system with no garbage
collection overhead~\cite{matsakis2014rust,rustbelt}.
Migrating existing C codebases to Rust would bring these guarantees to
code that has already been hardened and battle-tested over decades,
without discarding accumulated functionality and performance.
This is a compelling prospect for safety-critical systems where rewrites from
scratch are infeasible.

Automated translation from C to Rust is, however, technically difficult.
The closest off-the-shelf solution, \textsc{C2Rust}~\cite{c2rust}, applies
syntactic rewriting rules to produce a Rust program that is behaviorally
equivalent to the C source, but the output is \emph{wholly unsafe}: every
C pointer becomes a raw Rust pointer, and no ownership model is
synthesized.
The resulting code compiles and runs correctly, but it offers no
memory-safety guarantees beyond what C already provides.
Subsequent rule-based tools such as \textsc{Laertes}~\cite{laertes},
\textsc{Crown}~\cite{crown}, and \textsc{CRustS}~\cite{crusts} reduce
unsafe usage incrementally, but they are confined to the syntactic and
type-level patterns they were explicitly designed for; semantic
transformations such as replacing pointer arithmetic with iterators,
converting null-terminated \texttt{char*} to \texttt{\&str}, or introducing
\texttt{Result}-based error handling remain out of reach.

Large language models (LLMs) offer a path beyond pattern-based rewriting:
they can understand program intent and generate idiomatic Rust that
captures the semantics of C code in a human-readable, safe form.
Recent LLM-based translation systems~\cite{vert,sactor,rustmap,c2saferrust,evoc2rust,rustine}
have demonstrated impressive results on individual functions or small
benchmarks; they typically translate each function in isolation and
verify it independently, without assembling or testing the full codebase.
Scaling to real projects, however, exposes two fundamental challenges
that these approaches do not fully resolve.
Function-level LLM pipelines such as \textsc{C2SaferRust} suffer from
dependency propagation: any change to a function's type signature must
be coordinated across all callers, a cascading overhead that grows with
project scale. Project-scale systems such as \textsc{EvoC2Rust} attempt
to sidestep this constraint but produce output that does not compile on
any of the seven Coreutils programs we evaluate (functional correctness
0\% across all seven), because cross-unit dependencies remain unverified
until the full codebase is assembled. \textsc{Encrust} addresses both
failure modes.

\textbf{Challenge 1: decomposing translation complexity at ABI boundaries.}
Real C projects at scale involve deeply interleaved function dependencies:
functions share types, calling conventions, and global state across many
files.
A natural approach is to supply each function to the LLM together with its
full boundary context—call sites, global variable accesses, and
imports—so that the model can reason about how the translated function
will fit into its surroundings.
This context, however, substantially increases the difficulty of each
individual translation task: the LLM must simultaneously produce a
type-correct function body, a signature compatible with every call site,
and correct handling of any shared global state, all in a single
generation step.
The compounded complexity raises the per-function failure rate, depressing
the overall function compilance pass rate even when each sub-problem is
individually tractable.
What is needed is a mechanism that \emph{separates} the boundary
adaptation concern from the function logic concern, allowing the LLM to
focus on one at a time.

\textbf{Challenge 2: the limits of per-function scope.}
Per-function translation is insufficient in two respects:
\emph{verification scope} and \emph{transformation scope}.

\paragraph{Verification scope.}
Even when every individual translation compiles and passes its local
tests, the integrated crate may still fail: inter-unit dependencies
invisible at the per-function level surface as type mismatches,
unresolved symbols, or semantic divergences only when the full codebase
is assembled.
Correctness must therefore be verified on the whole integrated crate,
not on isolated translation units.

\paragraph{Transformation scope.}
Several categories of unsafe constructs are inherently intractable
within a single function's scope.
\texttt{Static mut} globals are accessed across many files simultaneously;
eliminating them safely requires reasoning about every read, write, and
address-of site across the entire program.
Similarly, unsafe constructs whose signatures are shared across many call
sites cannot be resolved without coordinated multi-file edits that a single
per-function generation step cannot produce.

Both respects call for a translation agent that operates on the
whole codebase, guided by an end-to-end verification gate.

\paragraph{\textsc{Encrust}.}
We present \textsc{Encrust}, a two-phase pipeline that addresses both
challenges.

\medskip
\noindent\textbf{Phase~1: Encapsulated Substitution} (\S\ref{sec:phase1}) resolves Challenge~1.
It translates each function using an \emph{ABI-preserving wrapper pattern}:
each function is split into a thin wrapper that preserves the original name
and C-compatible signature, and a safe inner function that implements the
translation.
The wrapper body consists exclusively of fixed \texttt{let}-binding
templates that perform raw-to-safe type conversions at each parameter
position, isolating all boundary-adaptation work in one dedicated shim.
Consequently, the LLM prompt for the safe inner function is stripped of
call-site obligations: it presents only the function's own logic, the safe
signatures of already-translated callees, and available safe struct
definitions, rather than the raw-pointer signatures at every call site or the
global variable layout.
This \emph{focused prompt} directly achieves the separation demanded by
Challenge~1: the LLM reasons about boundary conversion and function
logic in two separate, simpler steps rather than in one entangled
generation.
Untranslated callers continue to invoke the wrapper by its original name
with raw pointer types; a failed translation is silently rolled back and
the original unsafe body retained, so the crate remains compilable and
test-passing at every step (\emph{Live Scaffold Invariant}).
After all functions are processed, a deterministic
\emph{type-directed wrapper elimination} pass rewrites call sites to
invoke the safe inner functions directly, producing a wrapper-free crate.

\medskip
\noindent\textbf{Phase~2: Agentic Refinement} (\S\ref{sec:phase2}) resolves Challenge~2.
Residual unsafe patterns such as \texttt{static mut} globals involve
complex cross-file dependencies that no per-function pass can resolve in
isolation; addressing them requires whole-codebase visibility and
coordinated multi-file edits.
Phase~2 meets this requirement with an LLM agent equipped with 17
code-navigation and code-modification tools that operates on the entire
codebase, guided by a baseline-aware verification gate: no transformation
is committed unless the fully integrated crate compiles and passes the
complete test suite, catching the cross-unit failures that per-function
verification misses.

\medskip
We make the following contributions:

\begin{itemize}[noitemsep]
  \item \textbf{Encapsulated Substitution}
    (\S\ref{sec:func-trans}, \S\ref{sec:tdwe}):
    a translation scheme that splits each function into a thin
    \emph{caller-transparent wrapper} retaining the original name and
    raw-pointer signature, and a \emph{safe inner function} with
    idiomatic Rust types that the LLM targets with a focused,
    boundary-free prompt, enabling independent per-function translation
    with silent rollback and no coordinated caller updates, followed by
    a compile-and-test-gated \emph{type-directed wrapper elimination} (TDWE) pass.
  \item \textbf{An agentic refinement loop} (\S\ref{sec:phase2}): a
    multi-tool LLM agent with a task taxonomy, fixed processing order,
    and baseline-aware verification that resolves unsafe patterns
    requiring whole-codebase reasoning, plus the narrow set of
    unsafe-cast wrappers that TDWE pre-emptively deferred.
  \item \textbf{Evaluation on real-world benchmarks} (\S\ref{sec:experiment}):
    an assessment on 7 GNU Coreutils programs and 8 Laertes libraries
    (197,706 LoC, 2,366 functions total) showing that \textsc{Encrust}
    substantially reduces unsafe usage while preserving full end-to-end
    correctness across all 15 programs.
    Unlike prior work, which verifies only isolated translation units,
    \textsc{Encrust} validates the fully integrated crate; for the Laertes
    libraries, which ship without test cases, we construct harnesses and
    generate 1,000 inputs per library via AFL++ coverage-guided fuzzing
    to enable this end-to-end check for the first time.
\end{itemize}

%% ============================================================
%% Related Work: C-to-Rust Translation
%% ============================================================

\section{Related Work}
\label{sec:related}

%% ----------------------------------------------------------
\subsection{Rule-Based C-to-Rust Transpilation}
\label{sec:related:rule}
%% ----------------------------------------------------------

Early efforts to automate C-to-Rust migration rely on rule-based or
static-analysis-driven transpilation, translating C syntax and semantics
mechanically without large language models.

\textsc{C2Rust}~\cite{c2rust} is the most widely used starting point for automated C-to-Rust translation. It applies a fixed set of syntactic rewriting rules to produce Rust that is behaviorally equivalent to the input C, but the output is extensively unsafe: every C pointer becomes a raw pointer, and no ownership model is synthesized, meaning the generated code offers no memory-safety guarantees beyond what C already provides.

\textsc{Laertes}~\cite{laertes} builds on C2Rust output and uses the Rust compiler as a blackbox oracle, iteratively applying transformation rules and fixing type errors until the program compiles, thereby rewriting a subset of raw pointers to safe Rust references. It is restricted, however, to pointers that are unsafe solely due to missing ownership and lifetime information, excluding those involved in pointer arithmetic, unsafe casts, or other confounding factors. Semantic restructuring such as replacing pointer arithmetic with iterators, null-terminated char* with \&str, or C error codes with Result lies outside its scope.

\textsc{Crown}~\cite{crown} builds on C2Rust and applies ownership constraint analysis, using a SAT solver to infer whether each raw pointer is owning or non-owning, and retypes it accordingly to \texttt{Box} or a safe reference.

\textsc{CRustS}~\cite{crusts} also post-processes \textsc{C2Rust} output, applying 220 TXL source-to-source transformation rules to reduce unsafe keyword usage in function signatures and narrow the scope of unsafe blocks. Of these rules, 198 are strictly semantics-preserving and 22 are semantics-approximating, trading minor behavioral fidelity for greater safety. As with any rule-based approach, coverage is inherently limited to patterns anticipated at design time.

All of the above tools share a common limitation: they operate on syntactic or type-level patterns that can be identified without understanding program intent. Semantic transformations such as replacing pointer arithmetic with iterators, converting null-terminated \texttt{char*} to \texttt{\&str}, or introducing \texttt{Result}-based error handling require a level of comprehension that rule-based systems cannot provide. Our approach addresses this gap by using an LLM for semantic translation while retaining rule-based mechanisms for the systematic, deterministic aspects of type rewriting.

%% ----------------------------------------------------------
\subsection{LLM-Based C-to-Rust Translation}
\label{sec:related:llm}
%% ----------------------------------------------------------

LLM-based approaches split naturally into \emph{function-level} systems
that translate and verify each function in isolation, and
\emph{project-scale} systems that target whole repositories.
\textsc{Encrust} belongs to the latter group; we survey both.

\subsubsection{Function-Level Translation}
Early LLM-based work focuses on individual functions or small benchmarks.

\textsc{VERT}~\cite{vert} combines LLM-based and rule-based translation: it compiles the C program to WebAssembly and lifts it to Rust via rWasm to produce a semantically correct oracle, while in parallel using an LLM to generate a readable Rust candidate. The candidate is verified against the oracle using property-based testing and bounded model checking; if verification fails, a new candidate is regenerated.

\textsc{FLOURINE}~\cite{flourine} replaces the WebAssembly oracle with
differential fuzzing, checking input/output equivalence between the C
original and the translated Rust without requiring pre-written tests.

\textsc{Syzygy}~\cite{syzygy} combines LLM-driven code and test
translation guided by dynamic analysis, using execution information
to ensure that the generated Rust code and its tests are mutually
consistent and functionally equivalent to the original C.

\textsc{SmartC2Rust}~\cite{smartc2rust} segments the C program to fit within the LLM context window, then applies an iterative repair loop in which compilation errors, semantic discrepancies, and unsafe statements are each fed back to the LLM as separate repair signals until all errors are resolved.

\textsc{SACTOR}~\cite{sactor} is a LLM-driven C-to-Rust translation tool that employs a two-step pipeline: an initial unidiomatic translation that preserves C semantics, followed by an idiomatic refinement that eliminates unsafe blocks to conform to Rust standards. Static analysis guides both stages by providing the LLM with hints on pointer semantics and dependency resolution, and correctness is verified via FFI-based end-to-end testing.

\textsc{SafeTrans}~\cite{safetrans} is a framework that iteratively repairs compilation and runtime errors in two phases: a basic repair phase using compiler feedback, followed by a few-shot guided repair phase that provides error-type-specific context and code examples to the LLM. The repair strategies are informed by an analysis of which Rust features LLMs most frequently mistranslate.

\textsc{TymCrat}~\cite{tymcrat} focuses on type migration, i.e., replacing C types with appropriate idiomatic Rust types in function signatures. To handle the challenging many-to-many correspondence between C and Rust types, they introduce three techniques: generating candidates signatures, providing translated callee signatures as context to the LLM, and iteratively fixing type errors using compiler feedback.

\subsubsection{Project-Scale Translation}
Several systems target whole-project migration, the regime most
relevant to real-world adoption.

\textsc{RustMap}~\cite{rustmap} decomposes a C project into translation units ordered by usage dependencies, combining static analysis of syntactic structure with dynamic call graph profiling to determine translation order. Each unit is translated by an LLM with iterative repair driven by compilation and test feedback, and the translated units are composed into a runnable Rust program. A limitation is that correctness cannot be guaranteed for functions not exercised by the provided test cases.

\textsc{LLMigrate}~\cite{llmigrate} addresses LLM ``laziness'' (the tendency to omit large code sections when processing long contexts) by splitting modules into individual functions using Tree-sitter, translating each independently, and reassembling them in call graph topological order. It is validated on three Linux kernel modules: math, sort, and ramfs.

\textsc{EvoC2Rust}~\cite{evoc2rust} introduces a three-stage
skeleton-guided pipeline for project-level translation: it first
decomposes the C project into functional modules and uses a
feature-mapping-enhanced LLM to transform definitions and macros,
producing type-checked function stubs that form a compilable Rust
skeleton; it then incrementally fills in each stub with a full
function translation; finally, it applies a cascading repair chain
that combines rule-based transformations and LLM-driven refinement
to resolve remaining compilation errors.

\textsc{Rustine}~\cite{rustine} is a fully automated pipeline for
repository-level C-to-Rust translation that works from scratch rather
than building on \textsc{C2Rust} output. It applies program analysis
to refactor the C source prior to translation, resolving preprocessor
directives, minimizing pointer operations, and extracting caller-callee
and lifetime dependencies to guide translation order, and employs a
two-tier LLM strategy that escalates to a reasoning model only on
failure. Compared to six prior approaches, the resulting Rust code is
safer, more idiomatic, and more readable.

\textsc{C2SaferRust}~\cite{c2saferrust} combines \textsc{C2Rust} and
LLMs in a neuro-symbolic pipeline: it first uses \textsc{C2Rust} to
produce an initial unsafe Rust version, then decomposes it into
function-level slices ordered by static dependency analysis. Each
slice is augmented with dataflow and call graph context derived from
static analysis before being passed to an LLM for translation into
safer Rust, with compiler and test feedback driving iterative repair
until the slice compiles and passes end-to-end tests.

\textsc{IRENE}~\cite{irene} enhances LLM-based translation by
combining a rule-augmented retrieval module, which selects relevant
translation examples using syntactic rules from a static analyzer,
with a structured summarization module that generates a semantic
summary of the C code. Both outputs are composed into a single prompt,
and an error-driven module applies iterative compiler-feedback
refinement to correct translation errors.

\begin{figure}[htbp]
  \centering
  \includegraphics[width=\linewidth]{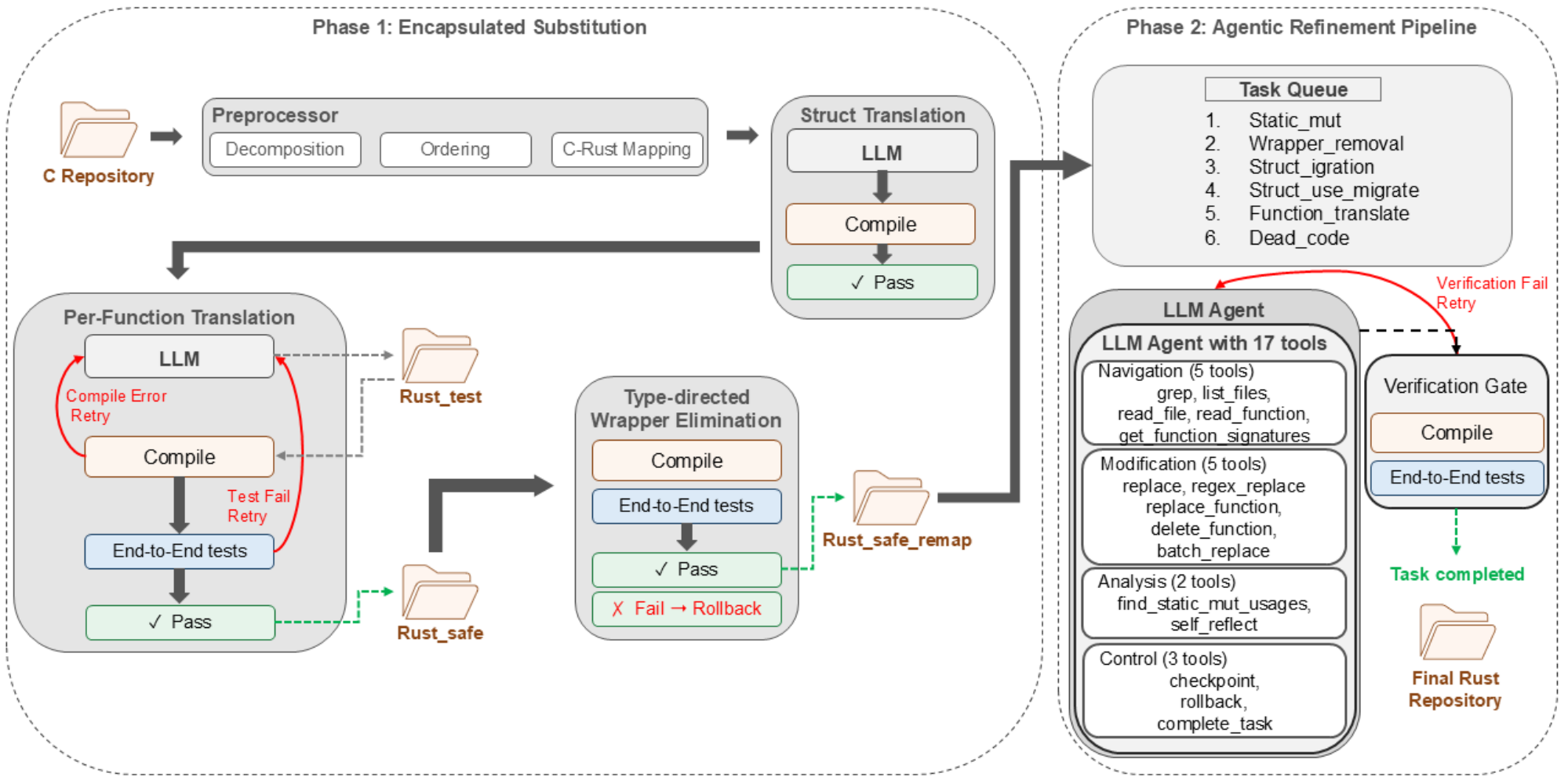}
  \caption{\textsc{Encrust} two-phase translation pipeline.
    \textbf{Phase~1} (left): for each function, the LLM generates a
    wrapper/safe-function pair verified through a compile-and-test loop;
    successfully translated pairs are collected into \texttt{rust\_safe/}.
    After all functions are translated, type-directed wrapper elimination
    rewrites call sites to invoke safe functions directly, yielding a
    wrapper-free crate \texttt{rust\_safe\_remap/}.
    \textbf{Phase~2} (right): an LLM agent equipped with 17 tools
    resolves the unsafe patterns deferred from Phase~1---\texttt{static
    mut} globals, unsafe-cast wrapper pairs, failed struct translations,
    and functions that exhausted the per-function retry budget---through
    an iterative loop governed by a compile-and-test verification gate.}
  \label{fig:pipeline-overview}
\end{figure}

\section{ENCRUST}
\subsection{Phase 1: Encapsulated Substitution}
\label{sec:phase1}

\textsc{Encrust} structures the translation as a two-phase pipeline
(Figure~\ref{fig:pipeline-overview}) whose design is governed throughout
by the \emph{coexistence constraint}: at every intermediate step, the
partially translated crate must compile and pass the full test suite,
because untranslated callers continue to depend on already-replaced
functions using their original raw-pointer signatures.

The C source is first transpiled to unsafe Rust by
\textsc{C2Rust}, yielding a behaviorally equivalent but
wholly unsafe base crate (\texttt{rust/}).
\textbf{Phase~1} (Figure~\ref{fig:pipeline-overview}, left) transforms \texttt{rust/} into a safe
equivalent through four ordered stages that operate on two parallel
working copies:
\texttt{rust\_test/} serves as the live compile-and-test scaffold, while
\texttt{rust\_safe/} accumulates verified safe function pairs.
Preprocessing (\S\ref{sec:preprocessing}) computes a leaf-first translation
order and builds LLM prompting context;
struct translation (\S\ref{sec:struct-trans}) generates safe counterpart
types for C-originated structs;
per-function translation (\S\ref{sec:func-trans}) replaces each function
with an ABI-preserving wrapper/safe-function pair verified through
compile-and-test;
and type-directed wrapper elimination (\S\ref{sec:tdwe}) rewrites all call
sites, producing a wrapper-free crate (\texttt{rust\_safe\_remap/}).

\textbf{Phase~2} (\S\ref{sec:phase2}) resolves the unsafe patterns deferred
from Phase~1: \texttt{static mut} globals, unsafe-cast wrapper pairs
pre-emptively skipped by TDWE, and failed translations, through an LLM
agent that navigates, edits, compiles, and tests the Phase~1 output in an
iterative loop.

Phase~1's design is governed throughout by the coexistence constraint.
Three choices follow directly: dual-struct preservation
(\S\ref{sec:struct-trans}), the wrapper pattern (\S\ref{sec:func-trans}),
and per-function failure isolation (\S\ref{sec:compile-test}). All are
expressed in the following invariant, which \textsc{Encrust} maintains
throughout Phase~1:

\begin{invariant}[Live Scaffold]
\label{inv:scaffold}
At every point during Phase~1, the \texttt{rust\_test/} crate compiles and passes
the full test-vector suite.
Concretely:
\begin{enumerate}[nosep,leftmargin=*]
  \item \textbf{Before any stage begins:} the \textsc{C2Rust}-generated crate
        compiles (ensured by \textsc{C2Rust} itself).
  \item \textbf{After each struct translation step:} the newly appended safe struct
        block and its conversion implementations compile; a failure is corrected by
        truncating the file to its pre-insertion length, restoring the prior state.
  \item \textbf{After each per-function translation step:} the wrapper/safe pair
        inserted into \texttt{rust\_test/} compiles and all test vectors pass;
        a persistent failure triggers a full-file snapshot rollback and the original
        unsafe body is retained, leaving the crate unchanged.
  \item \textbf{After type-directed wrapper elimination:} the remapped crate
        \texttt{rust\_safe\_remap/} is verified by rerunning the test-vector suite;
        any regression causes the wrapper to be restored, reverting to
        \texttt{rust\_safe/}.
\end{enumerate}
\end{invariant}

The four stages are described next: preprocessing, struct translation, per-function
translation, and type-directed wrapper elimination.
Concrete code examples using \texttt{gettext\_quote}
(\texttt{quotearg.c}) and the \texttt{quoting\_options} struct appear
from \S\ref{sec:struct-trans} onward, where the running example first
becomes relevant.

\subsubsection{Preprocessing}
\label{sec:preprocessing}

Before any LLM call is made, Phase~1 establishes the context every
subsequent stage depends on: a decomposition of the code into
translatable units, a leaf-first translation order, and a C-to-Rust
source mapping for prompt construction.

\paragraph{Decomposition.}
\textsc{Encrust} decomposes both source representations into function-level records.
From the C source, each function's name, source location, and body are extracted,
along with struct definitions used during struct translation (\S\ref{sec:struct-trans}).
From the C2Rust Rust output, each function record captures its source text, file span,
direct callees, referenced globals, \texttt{extern "C"} blocks, and \texttt{use} imports.

\paragraph{Translation ordering.}
\textsc{Encrust} orders functions leaf-first using the call graph derived from the Rust
decomposition: callees are translated before their callers so that each function's
safe callee signatures are available for prompt construction.
For genuine strongly connected components, an arbitrary order is chosen
(deterministically by function name); any resulting type mismatches surface
as compile errors in the per-function retry loop and are repaired there.

\paragraph{C-to-Rust source mapping.}
Since C2Rust preserves C function names verbatim, each Rust function record
is matched to its C source by name lookup.

\subsubsection{Struct Translation}
\label{sec:struct-trans}
% ============================================================

Before translating functions, \textsc{Encrust} generates safe counterparts for
\texttt{c2rust}-produced structs with raw pointer fields.
For each such struct $S$, replacing it in place would break all call sites the moment any field type
changed; \textsc{Encrust} instead \emph{preserves} $S$ unchanged and appends a
parallel safe struct $\hat{S}$ (Definition~\ref{def:ssa}) to the same source file as the \emph{dual-struct}
representation, so that safe code operates on $\hat{S}$ while the C ABI
boundary remains intact.

\paragraph{Struct Classification}
\label{sec:struct-classify}

Structs are classified into three mutually exclusive categories:

\begin{itemize}[leftmargin=*]
  \item \emph{System structs}: types whose definitions originate from
        OS or libc headers (\lstinline{FILE}, \lstinline{stat},
        \lstinline{passwd}, \lstinline{dirent}, etc.).
        These represent kernel-defined layouts that must remain
        \lstinline{#[repr(C)]}; generating a safe counterpart would serve
        no purpose since they cannot be reconstructed from safe Rust types.
        They are retained as-is.

  \item \emph{Pointer-free structs}: structs whose every field is a
        primitive, a fixed-size array, or another pointer-free struct.
        Rust's type system already provides full safety guarantees for such
        types; no transformation is needed.

  \item \emph{Translatable structs}: structs that originate from C
        source (not from anonymous \texttt{c2rust}-internal unions) and
        contain at least one raw pointer field.
        These are the sole target of safe struct abstraction generation.
\end{itemize}

\paragraph{Safe Struct Abstraction}
\label{sec:dual-struct}

For each translatable struct, the LLM prompt supplies two inputs:
(1)~the \texttt{c2rust}-generated unsafe Rust struct $S$ already present
in \texttt{rust\_test}, which defines the exact field names and raw types
the LLM must work with; and (2)~optionally, the original C struct source
as a \emph{semantic hint}, helping the LLM infer field intent
(\eg~whether a \lstinline{*mut c_uint} is a nullable scalar or a
length-bounded array).
The C source is context only; the transformation target is always $S$,
the unsafe Rust struct.

\begin{definition}[Safe Struct Abstraction]
\label{def:ssa}
Given a \texttt{c2rust}-generated unsafe Rust struct $S$ already present
in \texttt{rust\_test}, a safe struct abstraction for $S$ consists of
three components:
\begin{itemize}[nosep,leftmargin=*]
  \item A \emph{safe struct} $\hat{S}$: a \lstinline{pub struct} named
        in \texttt{CamelCase} whose fields replace the raw pointer types
        of $S$ with safe Rust counterparts: string pointers become
        \lstinline{Option<CString>}, length-paired pointers become
        \lstinline{Vec<T>}, nullable scalar pointers become
        \lstinline{Option<Box<T>>}, and primitives are unchanged.
  \item A \emph{materialisation function}
        $\mathit{from}_S : \&S \to \hat{S}$, realised as
        \lstinline{impl From<&S> for }$\hat{S}$, that \emph{copies}
        all pointer-reachable data of $S$ into fresh Rust-owned
        allocations (Invariant~\ref{inv:copy}).
  \item A \emph{projection function}
        $\mathit{to\_raw}_S : \&\hat{S} \to S$, realised as a
        \lstinline{to_raw(&self) -> S} method, that constructs a value
        of type $S$ whose pointer fields alias into $\hat{S}$'s
        allocations.
        The returned $S$ borrows from $\hat{S}$ and must not outlive it.
\end{itemize}
\end{definition}

\noindent
Safe code operates exclusively on $\hat{S}$; wrappers call $\mathit{from}_S$
at entry and $\mathit{to\_raw}_S$ at C-ABI boundaries (\S\ref{sec:wrapper}).

A representative safe struct abstraction for \lstinline{quoting_options}
is shown in the supplementary material.
When an array length cannot be reliably inferred from the C source
(no integer field is paired with the pointer in the struct), the LLM
is instructed to use a conservative fixed bound and annotate it with a
\texttt{// FIXME: inferred length} comment; if the compile-and-test loop
rejects the result, the struct translation is rolled back and the struct
is recorded in \texttt{failed\_structs} for Phase~2 struct migration
(\S\ref{sec:phase2-structs}).

\paragraph{Memory Ownership Discipline}
\label{sec:ownership}

The correctness of the materialisation function $\mathit{from}_S$ rests
on a non-obvious ownership constraint that arises specifically from the
incremental nature of Phase~1.  During translation, safe and unsafe functions coexist in the
same crate.  An untranslated function may hold a live reference to the
same raw struct $S$ that a translated function has just converted via
$\mathit{from}_S$.
If $\mathit{from}_S$ transferred ownership of a pointer field (e.g., via
\lstinline{CString::from_raw(raw.left_quote)}), the resulting
\lstinline{CString} would free the C-allocated memory when it is dropped
at the end of the safe function's scope.  Any subsequent access to
\lstinline{raw.left_quote} by the still-live untranslated function would
then constitute a use-after-free.
This hazard is especially acute for string literals: in C, string literals
reside in read-only memory segments; calling \lstinline{free} on them
causes an immediate trap.

We therefore require the following invariant, which is communicated
verbatim to the LLM in its prompt alongside counterexamples:

\begin{invariant}[Copy, Never Own]
\label{inv:copy}
Every $\mathit{from}_S$ implementation must \emph{copy} all
pointer-reachable data into fresh Rust-owned allocations.
Specifically:
\begin{enumerate}[nosep,leftmargin=*]
  \item \textbf{String fields:} use
        \lstinline{CStr::from_ptr(p).to_owned()}.
        \lstinline{CString::from_raw(p)} is prohibited: it claims
        ownership of C memory and frees it on drop.
  \item \textbf{Buffer fields:} use
        \lstinline{slice::from_raw_parts(p, n).to_vec()}.
        \lstinline{Vec::from_raw_parts(p, n, cap)} is prohibited.
  \item \textbf{Scalar pointer fields:} use \lstinline{Box::new(*p)} to
        copy the pointed-at value.
  \item \textbf{In \lstinline{to_raw()}:} return pointers via
        \lstinline{.as_ptr()} into the safe struct's allocations; never
        call \lstinline{Box::into_raw()} on a cloned value, as that leaks
        memory.  The caller must ensure the returned $S$ does not outlive
        $\hat{S}$.
\end{enumerate}
\end{invariant}

\noindent
\paragraph{Enforcement.}
The compile-and-test loop provides a dynamic safety net for violations of
Invariant~\ref{inv:copy}: a generated \texttt{from} implementation that
incorrectly transfers ownership will typically trigger a use-after-free or
double-free detectable at test time.
Because test-vector coverage is not exhaustive, enforcement is best-effort;
violations in untested code paths remain undetected.

\noindent With safe struct abstractions appended to \texttt{rust\_test/},
per-function translation can use $\hat{S}$ types in safe function bodies
while the original \texttt{\#[repr(C)]} structs remain intact at every
C-ABI boundary.

\subsubsection{Per-Function Translation}
\label{sec:func-trans}

Per-function translation processes each function in leaf-first order.
The central challenge is how to replace a function's implementation
without breaking the call sites of untranslated callers;
the \emph{wrapper pattern} resolves this.

\paragraph{The Wrapper Pattern}
\label{sec:wrapper}

A direct replacement of a function body would change its type signature,
breaking every call site that still passes raw C types.
To preserve callability across the translation boundary, \textsc{Encrust}
replaces the original function with two definitions:

\begin{definition}[Wrapper Pattern]
\label{def:wrapper}
Given an unsafe function $f$ with C-compatible signature
$f : T_1 \times \cdots \times T_n \to T_r$, the wrapper pattern
produces:
\begin{itemize}[nosep,leftmargin=*]
  \item $\texttt{f\_safe} : \hat{T}_1 \times \cdots \times \hat{T}_n \to \hat{T}_r$,
        the \emph{safe inner function}.
        It carries the suffix \texttt{\_safe} and accepts safe Rust types
        $\hat{T}_i$: C strings become \lstinline{&CStr}, length-paired
        pointers become slices (\lstinline{&[T]} / \lstinline{&mut [T]}),
        nullable pointers become \lstinline{Option<&T>}, and primitives are
        unchanged. It implements the function's logic without raw pointer
        operations wherever possible.
  \item $f : T_1 \times \cdots \times T_n \to T_r$,
        the \emph{wrapper function}.
        It preserves the original name, signature, and qualifiers
        (\texttt{unsafe}, \texttt{pub}, \texttt{extern "C"},
        \texttt{\#[no\_mangle]}, \texttt{\#[inline]}).
        Its body consists exclusively of type-conversion \texttt{let}
        bindings that shadow the original parameter names, followed by a
        single call to $\texttt{f\_safe}$.
\end{itemize}
\end{definition}

This decomposition has three essential properties:
\emph{ABI preservation} (the wrapper retains the original name, signature,
and calling convention, so all call sites compile unchanged);
\emph{safety isolation} (unsafe operations are confined to the wrapper's
simple let-bindings, while the safe function is maximally free of \texttt{unsafe});
and the \emph{name-shadowing property} (each parameter is re-bound under its
original identifier, enabling the TDWE pass (\S\ref{sec:tdwe}) to rewrite
call sites by textual substitution without alpha-renaming).

\FloatBarrier
\begin{lstlisting}[language=Rust,
  caption={Wrapper pattern for a function taking a C string and a
    pointer-to-integer.  Lines 13--14 show the name-shadowing property:
    the same identifier is reused for the converted value.},
  label={lst:wrapper-full}]
// (1) Safe inner function: safe Rust types
fn quotearg_safe(
    name:  &std::ffi::CStr,
    count: &mut libc::c_int,
) {
    let s = name.to_str().unwrap_or("");
    *count += s.chars()
               .filter(|c| *c == (@'"'@))
               .count() as libc::c_int;
}

// (2) Wrapper: identical signature to the c2rust original
#[no_mangle]
pub unsafe extern "C" fn quotearg(
    name:  *const libc::c_char,
    count: *mut libc::c_int,
) {
    // name-shadowing: same identifier, converted type
    let name  = std::ffi::CStr::from_ptr(name);
    let count = &mut *count;
    quotearg_safe(name, count)
}
\end{lstlisting}

\paragraph{Compile-Test Loop and Snapshot Rollback}
\label{sec:compile-test}

The LLM prompt for each function $f$ supplies the original C source,
the C2Rust Rust body, safe callee signatures, and safe struct definitions
$\hat{S}$; the LLM emits a pair conforming to Definition~\ref{def:wrapper}.
Once the candidate pair is emitted, \textsc{Encrust} inserts it into
\texttt{rust\_test} and attempts to compile the crate.
A compile failure triggers a retry: the compiler error message is
appended to the prompt and a new candidate is requested.
If compilation succeeds, \textsc{Encrust} runs the project's test-vector suite;
a test failure is likewise fed back as a repair signal.
To prevent a failed candidate from corrupting subsequent translations,
the target source file is snapshotted before each insertion and restored
on any persistent failure.
If a function cannot be translated within the retry budget (five attempts),
its original unsafe body is retained and translation proceeds
to the next function in order; the live scaffold invariant ensures the
crate continues to compile throughout.
We treat compile-and-test passage as the per-function correctness
criterion: a translation is accepted when the crate compiles without
errors and all test vectors pass, establishing behavioral compatibility
with the original C program on the provided inputs.

\subsubsection{Type-Directed Wrapper Elimination (TDWE)}
\label{sec:tdwe}

With every translated function represented as a wrapper/safe pair,
the wrappers are now redundant: call sites can invoke $f\_\mathit{safe}$
directly, giving every exported symbol a safe Rust signature and
eliminating the naming indirection and wrapper boilerplate.

\paragraph{Approach}

Wrapper elimination is tractable because the wrapper pattern
(\S\ref{sec:wrapper}) enforces a uniform structure: each wrapper body is
a flat sequence of let-bindings followed by a single tail call to
$f\_\mathit{safe}$, so the per-parameter conversions are already written
down and need not be re-derived.

The pass extracts these conversions from each wrapper's let-bindings
(using the last binding for re-bound parameters) and rewrites every
call site $f(a_1,\ldots,a_n)$ to
$f\_\mathit{safe}(\mathit{cvt}(a_1),\ldots,\mathit{cvt}(a_n))$.

When direct extraction is ambiguous (\eg~the wrapper body contains control
flow or the caller has already partially converted its arguments), the pass
applies a six-layer fallback strategy in order:
\begin{enumerate}[nosep,leftmargin=*]
  \item \textbf{Direct match}: if the argument type already equals the safe
        parameter type, it is passed through unchanged.
  \item \textbf{Cross-module struct}: detect \lstinline{*const S} arguments
        corresponding to cross-module struct types and synthesize a conversion
        using the struct's \texttt{From} implementation.
  \item \textbf{Null-pointer optimization}: \lstinline{null_mut()} arguments
        in already-translated callers become \lstinline{None} for
        \lstinline{Option}-typed parameters.
  \item \textbf{Safe-caller undo}: if the call site is inside an
        already-translated safe function, undo any prior safe-to-raw
        conversion applied to the argument.
  \item \textbf{Slice synthesis}: merge a \lstinline{(*mut T, size)}
        argument pair into \lstinline{\&mut [T]} via
        \lstinline{from\_raw\_parts\_mut}.
  \item \textbf{Type-level fallback}: synthesize a conversion from the raw
        and safe parameter types alone, without consulting the wrapper body.
\end{enumerate}
If none of the six layers produces a valid conversion for a given call site,
the call to $f$ is left unchanged and the wrapper is retained; the crate
remains compilable.

Listing~\ref{lst:tdwe} illustrates the pass on the running example.
\lstinline{gettext_quote} in \texttt{rust\_safe/} consists of a safe inner
function (\lstinline{gettext_quote_safe}) and a wrapper that accepts a raw
\lstinline{*const c_char} and converts it to \lstinline{&CStr} via
\lstinline{CStr::from_ptr}.
The remapping pass (1)~reads the single let-binding
\lstinline{let msgid = CStr::from_ptr(msgid)} from the wrapper body to
derive $\mathit{cvt}(\mathtt{msgid}) = \mathtt{CStr::from\_ptr(msgid)}$;
(2)~removes the wrapper; (3)~renames \lstinline{gettext_quote_safe} to
\lstinline{gettext_quote} and promotes it to \lstinline{pub}; and
(4)~rewrites every call site to apply the extracted conversion to each
raw-pointer argument.

\FloatBarrier
\begin{lstlisting}[language=Rust,
  caption={Type-directed wrapper elimination for \texttt{gettext\_quote}
    (GNU Coreutils \texttt{cat}).
    \emph{Before}: \texttt{rust\_safe/} contains a wrapper/safe pair
    (Definition~\ref{def:wrapper}); call sites pass raw \texttt{*const c\_char} pointers.
    \emph{After}: the wrapper is removed, \texttt{gettext\_quote\_safe} is promoted to
    \texttt{pub gettext\_quote}, and every call site is rewritten by applying
    and every call site is rewritten accordingly.},
  label={lst:tdwe}]
// ===== rust_safe/ (before elimination) =====

// Safe inner function (private)
fn gettext_quote_safe(msgid: &CStr, s: quoting_style) -> *const libc::c_char { /* ... */ }

// Wrapper: accepts raw pointer, converts, delegates
unsafe extern "C" fn gettext_quote(
    mut msgid: *const libc::c_char, mut s: quoting_style,
) -> *const libc::c_char {
    let msgid = CStr::from_ptr(msgid);   // cvt extracted here
    gettext_quote_safe(msgid, s)
}

// Call sites pass raw pointers directly to wrapper
left_quote  = gettext_quote(b"`\0" as *const u8 as *const libc::c_char, quoting_style);
right_quote = gettext_quote(b"'\0" as *const u8 as *const libc::c_char, quoting_style);

// ===== rust_safe_remap/ (after elimination) =====

// Wrapper removed; safe function promoted to public API
pub fn gettext_quote(msgid: &CStr, s: quoting_style) -> *const libc::c_char { /* ... */ }

// Call sites rewritten
left_quote  = gettext_quote(
    unsafe { CStr::from_ptr(b"`\0" as *const u8 as *const libc::c_char) }, quoting_style);
right_quote = gettext_quote(
    unsafe { CStr::from_ptr(b"'\0" as *const u8 as *const libc::c_char) }, quoting_style);
\end{lstlisting}

\paragraph{Compile-and-Test as the Universal Failure Gate}
\label{sec:tdwe-rollback}

TDWE attempts elimination on every wrapper/safe pair and relies on the
same compile-and-test gate that governs Phase~1 function translation
(\S\ref{sec:compile-test}) to decide whether each attempt succeeds.
Before attempting elimination of $f$, TDWE snapshots the affected source
files; after rewriting the call sites and deleting $f$, it compiles the
crate and runs the full test-vector suite.
If both stages pass, the elimination is committed; otherwise the snapshot
is restored and $f$'s wrapper is retained unchanged, preserving the Live
Scaffold Invariant.

This design is consistent with the system's broader correctness principle
and renders pre-emptive classification of ``hard'' cases unnecessary.
Cases that earlier wrapper-elimination approaches handled conservatively
all produce detectable failures under compile-and-test:
\emph{variadic signatures} cannot be given a safe Rust call-site rewrite
and produce a compile error;
\emph{complex return or parameter types} (\eg~\lstinline{Result<T,E>},
slice-of-references) cause type-mismatch errors at rewritten call sites;
\emph{chain dependencies} (a call site inside a not-yet-remapped wrapper)
produce either a type error or a test failure that the gate intercepts.
In every such case the rollback fires and the wrapper is retained without
any special-case logic.

TDWE processes wrappers in the same leaf-first order established by Phase~1
(\S\ref{sec:preprocessing}), so callees are remapped before their callers.
This ordering minimises chain-dependency conflicts without requiring explicit
deferral logic: when a caller's call site is rewritten, its callees have
already been remapped to safe signatures and the types are consistent.

\paragraph{Single Genuine Pre-emptive Skip: Unsafe-Cast Function-Pointer Use}
\label{sec:tdwe-skip}

One category requires a targeted pre-check because compile-and-test alone
cannot detect the failure.
When $f$ is referenced in an \emph{unsafe-cast} expression such as
\lstinline{f as *const ()}, \lstinline{f as *mut ()},
or \lstinline{f as unsafe extern "C" fn(...)}, the
cast discards $f$'s type, so renaming $f$'s signature from raw-pointer to
safe types leaves the cast expression syntactically valid and the crate
continues to compile.
The stored raw address, however, now refers to a function with a different
calling convention, producing undefined behaviour at the call site without
any compile error or test signal unless the unsafe-cast code path happens
to be exercised by the test suite.

For \emph{typed} function-pointer use (\eg~\lstinline{let fp: fn(*const c_char) = f}
or \lstinline{Some(f)}), a signature change produces a type-mismatch compile
error that the rollback gate intercepts correctly; no pre-check is needed.

TDWE therefore applies one targeted scan before attempting elimination:
it searches all source files for occurrences of the pattern
\lstinline{f as *const}/\lstinline{f as *mut}/\lstinline{f as unsafe}.
If any unsafe-cast site referencing $f$ is found, the pair is deferred
to Phase~2's \texttt{wrapper\_removal} task (\S\ref{sec:phase2-wrapper}),
which can retain a minimal \texttt{unsafe extern "C"} shim preserving
the C-ABI address while documenting the remaining unsafe contract.
In all other cases the wrapper is left in place only when the compile-and-test
gate fires, so the crate continues to compile and pass the test suite unchanged.

When the same wrapper/safe pair appears in two different source files
(a duplicate arising from multi-file crates), the pass designates one file as the \emph{canonical} module (chosen
deterministically by lexicographic file path for reproducibility) and
applies full remapping there; the duplicate module has its wrapper and safe
function replaced by a \lstinline{use} re-export of the canonical name,
preserving ABI visibility without duplicating code.

\noindent A final \texttt{reduce\_unsafe\_constructs} pass removes residual
\texttt{unsafe} block wrappers around calls that no longer require them.
Phase~1 thus delivers \texttt{rust\_safe\_remap/}: a crate that
compiles, passes the test suite, and is free of wrapper indirection.
Three categories of unsafe constructs remain deferred to Phase~2 by
their nature, plus a narrow fourth category arising from the single
pre-emptive skip:
\texttt{static mut} globals (require whole-program reasoning),
unsafe-cast wrapper/safe pairs deferred by TDWE's pre-check
(\S\ref{sec:tdwe-skip}), failed struct translations, and functions that
exhausted Phase~1's per-function retry budget.
Phase~2 (\S\ref{sec:phase2}) resolves these through six task types.
The four deferred categories map directly to \texttt{static\_mut},
\texttt{wrapper\_removal}, \texttt{struct\_migration}, and
\texttt{function\_translate} tasks; \texttt{struct\_use\_migrate} is a
follow-on that updates call sites after struct migration completes; and
\texttt{dead\_code} is a final cleanup sweep run after all translations
finish.

\subsection{Phase 2: Agentic Refinement}
\label{sec:phase2}

The four deferred categories from \S\ref{sec:tdwe} are addressed by
Phase~2 (Figure~\ref{fig:pipeline-overview}, right) with an LLM agent equipped with 17 tools.
Three of the four categories require
multi-step, whole-codebase reasoning that inherently exceeds Phase~1's
per-function scope: \texttt{static mut} globals, failed struct
translations, and functions that exhausted the retry budget.
The fourth, unsafe-cast wrapper pairs, was pre-emptively skipped by TDWE
because the compile-and-test gate cannot detect the resulting silent
undefined behaviour; Phase~2 handles these agentically with explicit
unsafe-cast site inspection.
Phase~2 addresses all four with an LLM agent equipped with 17 tools
(Table~\ref{tab:tools}), running an iterative loop that reads, edits,
compiles, and tests until each task passes a verification gate.
Because Phase~1 has already migrated most callers to safe types, Phase~2
transforms code \emph{in place} rather than introducing new wrapper pairs.
Intermediate results are persisted to disk so that an interrupted run can
be resumed without repeating completed tasks.

\paragraph{Task discovery and processing order.}
Phase~2 decomposes residual work into \emph{transformation tasks} and
processes them in a fixed type-level order:
\begin{enumerate}[nosep,leftmargin=*]
  \item \textbf{\texttt{static\_mut}} — one task per \texttt{static mut}
        declaration found by scanning the workspace.
  \item \textbf{\texttt{wrapper\_removal}} — one task per
        \texttt{f}/\texttt{f\_safe} pair that TDWE pre-emptively skipped
        due to an unsafe-cast site referencing~$f$
        (\S\ref{sec:tdwe-skip}), identified by rescanning the workspace
        for the pattern \lstinline{f as *const}/\lstinline{f as *mut}.
  \item \textbf{\texttt{struct\_migration}} — one task per struct that
        Phase~1 failed to translate.
  \item \textbf{\texttt{struct\_use\_migrate}} — one task per function
        whose signature references a raw struct type~$S$ for which
        \texttt{struct\_migration} just produced a safe counterpart~$\hat{S}$;
        updates the function's parameter types and all call sites from~$S$
        to~$\hat{S}$.  This task type always runs immediately after
        \texttt{struct\_migration} is complete.
  \item \textbf{\texttt{function\_translate}} — one task per function that
        exhausted Phase~1's retry budget, sorted leaf-first by call-graph
        order.
  \item \textbf{\texttt{dead\_code}} — a single cleanup task run last,
        after all translations are complete.
\end{enumerate}
The rationale for this order is as follows.
\texttt{static\_mut} elimination runs first so that subsequent tasks
never encounter unresolved global state.
\texttt{wrapper\_removal} runs second: the unsafe-cast wrappers it
handles are already written in terms of Phase~1-safe types, so their
call-site rewrites are independent of struct layout; resolving them
before struct migration avoids ambiguity between raw and safe struct
uses during call-site rewriting.
\texttt{struct\_migration} and \texttt{struct\_use\_migrate} then
complete the struct API before \texttt{function\_translate} begins,
ensuring that function translation tasks can assume a fully migrated
type environment.

\subsubsection{The Agentic Execution Loop}
\label{sec:phase2-loop}

Each task is handled by a single \emph{agent conversation}: an LLM
equipped with a task-specific system prompt and 17 tools, running in a
loop until it calls \texttt{complete\_task} or exhausts the iteration
budget (\texttt{MAX\_ITER} = 40).
The loop is governed by three mechanisms: a verification gate that enforces
behavioral correctness, checkpoints for safe rollback, and stall-detection
heuristics to escape repetitive tool-call patterns.

\paragraph{Tool suite.}
\label{sec:phase2-tools}
Table~\ref{tab:tools} lists the 17 tools grouped into five categories
(Navigation, Modification, Analysis, Verification, Control).

\begin{table}[htbp]
  \caption{Phase~2 tool suite (17 tools). All tools operate on the
           working copy of \texttt{rust\_safe\_remap/}.}
  \label{tab:tools}
  \centering
  \small
  \begin{tabular}{@{}l l p{5.2cm}@{}}
    \toprule
    \textbf{Category} & \textbf{Tool} & \textbf{Description} \\
    \midrule
    \multirow{5}{*}{Navigation}
      & \texttt{grep}                      & Regex search across the workspace; returns matching file:line pairs with surrounding context \\
      & \texttt{read\_file}                & Read file contents with an optional start/end line range \\
      & \texttt{read\_function}            & Extract a complete function body by brace-depth counting \\
      & \texttt{get\_function\_signatures} & List all function signatures declared in a given file \\
      & \texttt{list\_files}               & Enumerate all \texttt{.rs} source files in the workspace \\
    \midrule
    \multirow{5}{*}{Modification}
      & \texttt{replace}                   & Exact-match replacement of a unique string in a single file \\
      & \texttt{regex\_replace}            & Regex-based replacement with capture-group back-references \\
      & \texttt{replace\_function}         & Replace an entire function body identified by its name \\
      & \texttt{delete\_function}          & Delete a function together with its preceding attributes \\
      & \texttt{batch\_replace}            & Atomic multi-file replacement applied in a single operation \\
    \midrule
    \multirow{2}{*}{Analysis}
      & \texttt{find\_static\_mut\_usages} & Classify all read, write, and address-of sites of a \texttt{static mut}; performs signal-reachability analysis \\
      & \texttt{self\_reflect}             & Count remaining unsafe blocks, unsafe functions, and raw pointer occurrences \\
    \midrule
    \multirow{2}{*}{Verification}
      & \texttt{compile}                   & Run \texttt{cargo build} in debug or release mode; return compiler diagnostics on failure \\
      & \texttt{run\_tests}                & Build the release binary and execute the full test-vector suite \\
    \midrule
    \multirow{3}{*}{Control}
      & \texttt{checkpoint}                & Snapshot all \texttt{.rs} files and \texttt{Cargo.toml} to a named restore point \\
      & \texttt{rollback}                  & Restore the workspace to a previously saved named snapshot \\
      & \texttt{complete\_task}            & Signal task completion; triggers the two-stage verification gate \\
    \bottomrule
  \end{tabular}
\end{table}

\paragraph{Verification gate.}
\label{sec:phase2-verify}
The LLM signals completion by calling \texttt{complete\_task(summary)}.
This call does \emph{not} immediately mark the task done; the loop
intercepts it and runs a two-stage gate: (1)~\texttt{cargo build}
(debug): on failure the compile errors are returned as a tool result
and the LLM must fix them before calling \texttt{complete\_task} again;
(2)~\texttt{cargo build --release} followed by all test vectors: on
failure the failing tests are returned.
Only when both stages pass is the task marked \textsc{Completed}.
The gate additionally applies \emph{baseline-aware acceptance}: before
Phase~2 begins, the orchestrator records which test vectors pass on the
unmodified Phase~1 output.
If stage~(2) fails but every failing test also failed in that baseline,
the gate treats the outcome as non-regressing and still marks the task
\textsc{Completed}, preventing pre-existing Phase~1 failures from
permanently blocking otherwise correct transformations.

\begin{claim}[Verification Gate Soundness]
\label{claim:verify-gate}
A task is marked \textsc{Completed} only when \texttt{cargo build}
succeeds and no test vector that passed in the Phase~1 baseline is
caused to fail.  No completed task introduces a new compilation error
or a test regression relative to the Phase~1 baseline.
\end{claim}

\noindent\textit{Proof.} By construction of the two-stage gate
(\S\ref{sec:phase2-verify}).\smallskip

\paragraph{Checkpoint and rollback.}
\label{sec:phase2-checkpoint}
Before each task begins, the orchestrator snapshots the workspace state.
The LLM can roll back to this snapshot at any point to discard partial
edits and restart from a clean state.
If a task exhausts its iteration budget without passing the verification
gate, the orchestrator automatically restores the pre-task snapshot,
ensuring failed tasks do not corrupt the workspace for subsequent tasks.
Completed tasks are never rolled back; subsequent tasks therefore see
exactly the cumulative effect of all previously completed work.

\paragraph{Stall recovery.}
After five consecutive \texttt{compile} failures, the loop injects a
task-specific hint (try a simpler type, rollback, or retain
\texttt{unsafe fn} with a \texttt{// SAFETY:} comment).
Cyclic tool-call patterns are detected via a sliding window of
(tool, argument-hash) fingerprints: period-1 spins
(same call repeated), period-2 alternations ($A,B,A,B,\ldots$), and
period-3 cycles ($A,B,C,A,B,C,\ldots$) each trigger a diagnostic hint
and clear the window.

The complete pseudocode is given in the supplementary material.

\subsubsection{Static Mut Elimination}
\label{sec:phase2-static}

\texttt{static mut} declarations in \textsc{C2Rust} output represent
C global variables; every access requires an \texttt{unsafe} block and
is classified as undefined behaviour without additional synchronisation.
One task is created per unique variable found by scanning the workspace.

\paragraph{Signal-reachability.}
\texttt{find\_static\_mut\_usages} determines whether a variable is
reachable from a signal handler by BFS from registered handlers over
an approximate call graph, and annotates each usage site with a
\texttt{signal\_reachable} flag.

\paragraph{Safe replacement strategy.}
The LLM calls \texttt{find\_static\_mut\_usages("NAME")} to classify all
reads, writes, and address-of sites, then selects the simplest safe
pattern from the following priority-ordered rules:
\begin{enumerate}[nosep,leftmargin=*]
  \item \emph{No writes, compile-time value} $\to$ \lstinline{const X: T = val}
  \item \emph{No writes, runtime-init value} $\to$ \lstinline{static X: T = val} (non-mut)
  \item \emph{Single write, init pattern} $\to$ \lstinline{static X: OnceLock<T>}
  \item \emph{Scalar, any writes} $\to$ \lstinline{static X: AtomicT}
  \item \emph{Complex / non-scalar} $\to$ \lstinline{static X: OnceLock<Mutex<T>>}
\end{enumerate}
Signal-reachable variables are restricted to rules~1, 2, and~4;
\texttt{OnceLock} and \texttt{Mutex} are excluded because they are not
async-signal-safe.
After selecting a safe type, the LLM applies \texttt{batch\_replace} to
update the declaration and every usage site across all files, then
iterates with \texttt{compile} until the crate builds cleanly.

Listing~\ref{lst:atomic-example} shows a representative transformation.

\FloatBarrier
\begin{lstlisting}[language=Rust,
  caption={Before and after Phase~2 \texttt{static\_mut} elimination for
    a counter accessed from signal-handler context.},
  label={lst:atomic-example}]
// BEFORE: c2rust output
static mut lines_printed: libc::c_int = 0 as libc::c_int;
unsafe { lines_printed += 1; }
unsafe { if lines_printed > limit { ... } }

// AFTER: LLM chose AtomicI32 (signal-reachable)
use std::sync::atomic::{AtomicI32, Ordering};
static LINES_PRINTED: AtomicI32 = AtomicI32::new(0);
LINES_PRINTED.fetch_add(1, Ordering::SeqCst);
if LINES_PRINTED.load(Ordering::SeqCst) > limit { ... }
\end{lstlisting}

\subsubsection{Wrapper Removal}
\label{sec:phase2-wrapper}

TDWE eliminates all wrapper/safe pairs except those pre-emptively skipped
due to unsafe-cast function-pointer use (\S\ref{sec:tdwe-skip}).
Phase~2's \texttt{wrapper\_removal} task targets only this narrow residual
set, discovered by rescanning the workspace for \texttt{f}/\texttt{f\_safe}
pairs where at least one \lstinline{f as *const}/\lstinline{f as *mut}
site triggered TDWE's pre-emptive skip.

\paragraph{Removal strategy.}
The LLM inspects every unsafe-cast site and determines whether the
C-ABI address of $f$ must remain stable (\eg~the pointer is passed to
an external C callback registration such as \lstinline{atexit} or
\lstinline{signal}) or whether the cast is incidental and can be
rewritten to use the renamed safe function directly.
If the address must be preserved, the LLM retains a minimal
\texttt{unsafe extern "C"} shim for $f$ and inlines the safe logic into
\texttt{f\_safe}, documenting the retained contract with a
\texttt{// SAFETY:} comment.
If the cast can be rewritten, the LLM updates all unsafe-cast sites,
rewrites direct call sites, deletes the wrapper, and renames
\texttt{f\_safe} to \texttt{f} restoring its original visibility and ABI
attributes, then verifies through the standard compile-and-test gate.

\subsection{Struct Migration and Use Update}
\label{sec:phase2-structs}
\label{sec:phase2-struct-use}

Phase~2 generates one \texttt{struct\_migration} task per struct that
Phase~1 failed to translate, following the same strategy as
\S\ref{sec:struct-trans}: replace pointer fields with owned types,
implement \texttt{From<\&RawS>} (copying, never owning) and
\texttt{to\_raw()}, and append the safe struct in the same file.

Once all \texttt{struct\_migration} tasks complete, one
\texttt{struct\_use\_migrate} task is run per function whose signature
still references a raw struct type $S$ for which a safe counterpart $\hat{S}$
was just produced.
The LLM updates the function's parameter types from $S$ to $\hat{S}$,
rewrites the body to operate on $\hat{S}$ fields directly (using
\texttt{to\_raw()} only at C-ABI boundaries), and propagates call-site
updates across the workspace.
Both task types complete before \texttt{function\_translate} begins,
ensuring a fully migrated type environment for subsequent translation.

\subsubsection{Function Translation}
\label{sec:phase2-functions}

\paragraph{In-place transformation.}
Unlike Phase~1, Phase~2 transforms failed functions \emph{in place}
(no new wrapper pair is introduced), since most callers have already
been migrated to safe types:
\begin{itemize}[nosep,leftmargin=*]
  \item Remove \texttt{unsafe} from the signature (keep \texttt{pub},
        \texttt{extern "C"}, \texttt{\#[no\_mangle]} as present).
  \item Replace raw pointer parameters with safe Rust types
        (\lstinline{*const c_char}$\to$\lstinline{&CStr},
         \lstinline{*mut T}$\to$\lstinline{&mut T}, etc.).
  \item Rewrite the body with safe idioms; wrap unavoidable unsafe
        operations in minimal \lstinline|unsafe { // SAFETY: ... }|
        blocks.
  \item Apply \texttt{batch\_replace} to update all call sites across
        the workspace.
\end{itemize}
This avoids introducing new wrapper pairs that would require further
cleanup.

\paragraph{Translation order.}
Tasks are sorted leaf-first by call-graph order: a function is scheduled
only after its failed callees have been resolved, so each translation can
assume its callees already have safe signatures.

\paragraph{Workflow.}
The LLM reads the function and its call sites, rewrites the function
in place, iterates compile until clean, then signals completion.
For genuine FFI boundaries that cannot be made safe, the LLM retains
\texttt{unsafe fn} with a \texttt{// SAFETY:} comment.

\subsubsection{Dead Code Cleanup}
\label{sec:phase2-dead}

A single \texttt{dead\_code} task runs after all other task types are
complete.
The LLM identifies and removes dead artifacts produced by earlier
translation stages: wrapper functions whose paired safe function was
subsequently inlined or renamed, unused safe structs and their conversion
impls, unused imports, and empty \texttt{extern "C" \{\}} blocks.
Each batch of removals is followed by recompilation to confirm no
regressions.
Formal correctness claims (behavioral preservation and ordering soundness)
are stated in the supplementary material.

\section{Experiments}
\label{sec:experiment}

% ---------------------------------------------------------------
% NARRATIVE SPINE:
%   §3.1 Benchmarks    — what we tested on
%   §3.2 Metrics       — how we measured
%   §3.3 Setup         — model / compute / reproducibility
%   §3.4 RQ1           — overall unsafe reduction (vs. baselines)
%   §3.5 RQ2           — Phase 1 effectiveness (translation rate + TDWE)
%   §3.6 RQ3           — Phase 2 effectiveness (per-task breakdown)
%   §3.7 Limitations
% ---------------------------------------------------------------

% ---------------------------------------------------------------
\subsection{Benchmarks}
\label{sec:benchmarks}
% ---------------------------------------------------------------

We evaluate on two benchmark suites comprising 15 real-world C programs
in total (Table~\ref{tab:coreutils-stats} and Table~\ref{tab:laertes-stats}).
Both suites are used by prior C-to-Rust translation
work~\cite{c2saferrust,laertes,crown}, enabling direct comparison.
A key property of our evaluation is that \emph{all 15 programs} are
evaluated with correctness verification: unlike prior work on the Laertes
suite, which reports safety metrics only because the libraries ship without
test cases, we construct test harnesses and collect inputs via grey-box
fuzzing for the Laertes libraries, enabling end-to-end behavioral
validation on the full benchmark.

\paragraph{GNU Coreutils.}
We use the 7-program GNU Coreutils benchmark introduced by
Nitin et al.~\cite{c2saferrust}.
The programs range from 5,859 (\texttt{pwd}) to 14,423 (\texttt{tail})
lines of C, totalling 65,117 LoC and 1,210 functions, of which 314 are
exercised by the test suite.
Each program ships with 2--30 shell test scripts (63 scripts total) that
were constructed using coverage-guided fuzzing by the benchmark
authors~\cite{c2saferrust}; we use these scripts directly as our
correctness oracle.

\begin{table}[htbp]
  \centering
  \small
  \caption{GNU Coreutils benchmark statistics~\cite{c2saferrust}.
           ``Covered'' = functions exercised by the test suite.}
  \label{tab:coreutils-stats}
  \begin{tabular}{lrrrr}
    \toprule
    Program & LoC & Functions & Covered & Test scripts \\
    \midrule
    \texttt{split}    & 13,848 & 207 &  73 & 12 \\
    \texttt{pwd}      &  5,859 & 127 &  16 &  2 \\
    \texttt{cat}      &  7,460 & 166 &  37 &  4 \\
    \texttt{truncate} &  7,181 & 124 &  33 &  8 \\
    \texttt{uniq}     &  8,299 & 167 &  34 &  3 \\
    \texttt{tail}     & 14,423 & 266 &  76 & 30 \\
    \texttt{head}     &  8,047 & 153 &  45 &  4 \\
    \midrule
    \textbf{Total}    & 65,117 & 1,210 & 314 & 63 \\
    \bottomrule
  \end{tabular}
\end{table}

\paragraph{Laertes benchmark.}
We additionally evaluate on 8 of the 10 libraries in the Laertes
benchmark~\cite{laertes}, excluding \texttt{xzoom} and \texttt{grabc}
because their graphical interfaces preclude automated test execution in a
headless environment.
The remaining 8 libraries range from 49 (\texttt{qsort}) to 106,123
(\texttt{optipng}) lines of C, totalling 132,589 LoC and 1,156 functions.

The Laertes libraries ship without test cases.
To enable correctness verification and to match the end-to-end
evaluation methodology of the Coreutils suite, we write a lightweight
driver harness for each library that feeds arbitrary byte sequences from
the fuzzer to the library's public API, then collect test inputs using
\textsc{AFL++}~\cite{aflplusplus}.
Each candidate input is first executed against the original C binary;
inputs that crash the C binary or produce non-deterministic output across
two identical runs are discarded, retaining only inputs for which the C
binary yields a stable, reproducible output to serve as a ground-truth
oracle.
From the surviving inputs we select 1,000 per library.
The resulting inputs serve both as the correctness oracle for Phase~1's
compile-test loop and as the baseline for Phase~2's verification gate.

\begin{table}[htbp]
  \centering
  \small
  \caption{Laertes benchmark statistics~\cite{laertes}.
           Test inputs are collected via AFL++ fuzzing (1,000 per library).}
  \label{tab:laertes-stats}
  \begin{tabular}{lrrrr}
    \toprule
    Library & LoC & Functions & Covered & Test inputs \\
    \midrule
    \texttt{tulipindicators} & 12,565 & 270 & 211 & 1,000 \\
    \texttt{optipng}         & 106,123 & 495 & 224 & 1,000 \\
    \texttt{bzip2}           &  4,417 &  65 &  31 & 1,000 \\
    \texttt{snudown}         &  5,426 & 140 &  72 & 1,000 \\
    \texttt{lil}             &  3,565 & 148 &  72 & 1,000 \\
    \texttt{genann}          &    410 &  13 &   7 & 1,000 \\
    \texttt{urlparser}       &     34 &  21 &  10 & 1,000 \\
    \texttt{qsort}           &     49 &   4 &   4 & 1,000 \\
    \midrule
    \textbf{Total}           & 132,589 & 1,156 & 631 & 8,000 \\
    \bottomrule
  \end{tabular}
\end{table}

% ---------------------------------------------------------------
\subsection{Metrics}
\label{sec:metrics}
% ---------------------------------------------------------------

We measure translation quality along four dimensions.

\paragraph{Safety metrics.}
To enable direct comparison with C2SaferRust~\cite{c2saferrust}, we
adopt their five safety metrics:
\emph{raw pointer declarations} (\texttt{*const T} / \texttt{*mut T} type annotations),
\emph{raw pointer dereferences} (\texttt{*<expr>} occurrences in executable code),
\emph{unsafe lines of code} (total source lines enclosed in \texttt{unsafe\{\}} blocks),
\emph{unsafe type casts} (\texttt{transmute} calls and unsafe coercions), and
\emph{unsafe call expressions} (function calls that appear inside unsafe blocks).
The five metrics capture distinct facets of unsafety: a tool may
eliminate unsafe blocks while leaving raw pointer declarations intact,
or reduce dereferences without removing the enclosing unsafe scope.
Reporting all five prevents misleading single-metric summaries and allows
fine-grained comparison.
For each program we report the absolute count under each metric before
and after translation and compute the percentage reduction relative to the
\textsc{C2Rust} baseline.

\paragraph{Function compilance pass rate.}
We define the \emph{function compilance pass rate} as the fraction of
all functions whose LLM-based translation compiled successfully within
the tool's retry budget:
\[
  \text{Function comiliance pass rate} =
    \frac{\text{functions whose translation compiled}}
         {\text{total functions}}.
\]
A function counts as passing if the translated code compiles; it counts
as failing if every attempt results in a compilation error and the
tool retains or falls back to the original unsafe body.
The function compilance pass rate is applicable to any LLM-based translation
tool that attempts per-function translation with compilation verification,
including both \textsc{Encrust} and \textsc{C2SaferRust}, which both perform
incremental function-level translation with compile-and-test checking to
maintain functional equivalence throughout the process.
This metric measures the compile success rate of LLM-driven translation,
independent of whether the resulting code passes test vectors
(which is captured separately by the functional correctness metric).
Crucially, the function compilance pass rate is also independent of the five
safety metrics: a function whose unsafe body is retained due to a failed
translation may still contribute to safety metric reduction through other
transformations such as struct migration.

\paragraph{Functional correctness.}
We verify end-to-end behavioral equivalence between the translated Rust
binary and the original C program using two complementary counts.
The \emph{test-script pass rate} measures the fraction of test
scripts that pass: for the Coreutils programs, these are the 63 shell
scripts provided by the benchmark~\cite{c2saferrust}; for the Laertes
libraries, we count the fraction of the 1,000 AFL++-generated test inputs
on which the translated binary agrees with the C oracle.
The \emph{test-vector pass rate} reports the same agreement at the
granularity of individual input/output pairs, which is the unit used by
Phase~1's compile-test loop and Phase~2's verification gate.
By the Live Scaffold Invariant (Invariant~\ref{inv:scaffold}), both
measures must equal 100\% at every intermediate step and in the final
output; any deviation constitutes a correctness regression.
Reporting functional correctness for \emph{both} suites is a direct
advantage over prior work: C2SaferRust~\cite{c2saferrust},
Laertes~\cite{laertes}, and Crown~\cite{crown} cannot provide correctness
verification for the Laertes libraries because those libraries ship without
test inputs.

\paragraph{Idiomatic Rust quality.}
Beyond eliminating unsafe constructs, a high-quality C-to-Rust translation
should produce code that follows idiomatic Rust conventions.
We measure \emph{idiomatic quality} by the number of \textsc{Clippy} warnings
emitted on the final translated crate:
\[
  \text{Clippy warnings} = \texttt{cargo clippy 2>\&1 | grep -c "\^{}warning"}.
\]
\textsc{Clippy} is the official Rust linter and flags patterns that compile
correctly but deviate from community best practices, including redundant
clones, inefficient iterator usage, unnecessary unsafe blocks that wrap
safe operations, and opportunities to replace raw-pointer patterns with
idiomatic slice or reference APIs.
Fewer warnings indicate that the translated code is not merely safe but
also conforms to the style expected by Rust developers, which directly
affects maintainability.
A lower Clippy warning count is better; a count of zero means
\textsc{Clippy} raises no objections against the translated codebase.
Unlike the five safety metrics, the Clippy warning count is independent of
whether \texttt{unsafe} blocks are present: safe but non-idiomatic code
(e.g., manual index loops instead of iterators) also raises warnings.
We report this metric for \textsc{Encrust}'s final output and for each
baseline to characterise the idiomaticity gap between approaches.

% ---------------------------------------------------------------
\subsection{Experimental Setup}
\label{sec:setup}
% ---------------------------------------------------------------

\paragraph{LLM configuration.}
All translation steps in both phases use GPT-4o.
Phase~1 allows up to five compile-test retries per function: on each
retry the LLM receives the compiler error or test-failure output as
additional context and regenerates the wrapper/safe-function pair.
If all five attempts fail, the pipeline rolls back to the original unsafe
body and advances to the next function.
Phase~2 caps the agentic loop at 40 tool-call iterations per task; if the
loop reaches the cap without the verification gate passing, the workspace
is rolled back to the pre-task snapshot and the task is marked failed.
All LLM calls use temperature~$= 0.0$ to make the pipeline deterministic
conditioned on the model's weights and prompt.

\paragraph{Baselines.}
We compare \textsc{Encrust} against three baselines, all run on the same
Coreutils and Laertes programs under identical conditions.
\textbf{(B1)~\textsc{C2Rust}}~\cite{c2rust}: the raw transpiler output
with no further safety work; this serves as the lower bound and defines
the safety-metric baselines from which all tools reduce.
\textbf{(B2)~\textsc{C2SaferRust}}~\cite{c2saferrust}: a neuro-symbolic
LLM-based pipeline evaluated on the same two benchmark suites with the
same five safety metrics, enabling direct numeric comparison.
\textbf{(B3)~\textsc{EvoC2Rust}}~\cite{evoc2rust}: a skeleton-guided,
project-scale LLM pipeline that first generates compilable Rust stubs for
the entire project and then fills each stub with a full function
translation, making it the most architecturally comparable large-scale
baseline.
To ensure a controlled comparison, we run both \textsc{C2SaferRust} and
\textsc{EvoC2Rust} with GPT-4o at temperature~$= 0.0$, matching the
configuration used by \textsc{Encrust}.
Unlike \textsc{C2Rust} and \textsc{C2SaferRust}, \textsc{EvoC2Rust}
translates C directly to Rust without relying on \textsc{C2Rust} as an
intermediate; when it fails to translate a function, that function is
absent from the output rather than present as unsafe code.
Measuring the five safety metrics on an incomplete program would therefore
undercount unsafe constructs and artificially inflate \textsc{EvoC2Rust}'s
apparent safety score.
To ensure a fair comparison, we apply the following normalisation: for
every function that \textsc{EvoC2Rust} fails to produce, we substitute
the \textsc{C2Rust}-transpiled unsafe body of that function.
Safety metrics are then measured on this completed, compilable crate,
so that all three baselines and \textsc{Encrust} are evaluated on programs
with the same set of functions.

We evaluate \textsc{Encrust} along two dimensions.
(1) Overall safety improvement covers Phase~1 translation coverage
(function compilance pass rate), functional correctness preservation, and
unsafe-construct reduction relative to baselines across five safety metrics,
reported jointly in \S\ref{sec:safety}. And
(2) the marginal contribution of Phase~2 agentic refinement, examined as an
ablation (\S\ref{sec:phase2-contribution}).
All experiments use the benchmarks and metrics defined in \S\ref{sec:benchmarks}
and \S\ref{sec:metrics}.

\begin{table*}[htbp]
\centering
\caption{Safety and correctness metrics on the \textbf{GNU Coreutils} benchmark.
  Columns (all $\downarrow$ lower is better unless noted):
  \emph{Ptr.Decl}~=~raw pointer declarations;
  \emph{Ptr.Deref}~=~raw pointer dereferences;
  \emph{Unsafe LoC}~=~unsafe lines of code;
  \emph{Unsafe Cast}~=~unsafe type casts;
  \emph{Unsafe Call}~=~unsafe call expressions;
  \emph{Comp.Rate}~($\uparrow$)~=~fraction of LLM-translated functions that compiled
  (N/A for systems without LLM function translation);
  \emph{Correct.}~($\uparrow$)~=~test-script pass rate;
  \emph{Clippy}~=~\texttt{cargo clippy} warning count.
  \textsc{EvoC2Rust} Correct.~=~0\% and Clippy~=~N/A because the translated
  project does not compile.
  Best value per metric per program \textbf{bolded}; second-best \underline{underlined}.}
\label{tab:safety-coreutils}
\resizebox{\textwidth}{!}{%
\begin{tabular}{ll rrrrr r r r}
\toprule
Program & System & Ptr.Decl $\downarrow$ & Ptr.Deref $\downarrow$ & Unsafe LoC $\downarrow$ & Unsafe Cast $\downarrow$ & Unsafe Call $\downarrow$ & Comp.Rate $\uparrow$ & Correct. $\uparrow$ & Clippy $\downarrow$ \\
\midrule
\multirow{4}{*}{cat}
 & \textsc{C2Rust}      & 192 & 317 & 5{,}625 & 3{,}116 & 1{,}038 & N/A    & \textbf{100\%} & 546 \\
 & \textsc{C2SaferRust} & 120 & 240 & 4{,}189 & 1{,}992 &   \underline{912}   & \underline{75.3\%} & \textbf{100\%} & \underline{239} \\
 & \textsc{EvoC2Rust}   &  \textbf{81} & \underline{231} & \textbf{2{,}338} & \textbf{1{,}082} & \textbf{525} & 64\% & 0\% & N/A \\
 & \textsc{Encrust}     &  \underline{95} & \textbf{170} & \underline{3{,}407} & \underline{1{,}463} & 1{,}076 & \textbf{95.8\%} & \textbf{100\%} & \textbf{211} \\
\cmidrule(l){1-10}
\multirow{4}{*}{head}
 & \textsc{C2Rust}      & 192 & 442 & 6{,}245 & 3{,}488 & 1{,}378 & N/A    & \textbf{100\%} & 410 \\
 & \textsc{C2SaferRust} & 141 & 351 & 4{,}699 & 2{,}464 & 1{,}189 & \underline{69.9\%} & \textbf{100\%} & \underline{247} \\
 & \textsc{EvoC2Rust}   &  \textbf{65} & \underline{309} & \underline{3{,}525} & \underline{1{,}799} & \textbf{902} & 59\% & 0\% & N/A \\
 & \textsc{Encrust}     & \underline{109} & \textbf{155} & \textbf{2{,}244} & \textbf{806} & \underline{1{,}071} & \textbf{97.1\%} & \textbf{100\%} & \textbf{136} \\
\cmidrule(l){1-10}
\multirow{4}{*}{pwd}
 & \textsc{C2Rust}      & 164 & 295 & 4{,}201 & 2{,}248 &   875 & N/A    & \textbf{100\%} & 370 \\
 & \textsc{C2SaferRust} & 129 & 225 & 3{,}151 & 1{,}563 &   871 & \underline{70.9\%} & \textbf{100\%} & \underline{204} \\
 & \textsc{EvoC2Rust}   &  \textbf{62} & \underline{178} & \underline{1{,}756} &  \underline{575} & \textbf{548} & 68\% & 0\% & N/A \\
 & \textsc{Encrust}     &  \underline{80} & \textbf{49} & \textbf{1{,}180} & \textbf{228} &   \underline{748} & \textbf{99.2\%} & \textbf{100\%} & \textbf{108} \\
\cmidrule(l){1-10}
\multirow{4}{*}{split}
 & \textsc{C2Rust}      & 252 & 656 & 11{,}324 & 5{,}979 & 2{,}353 & N/A    & \textbf{100\%} & 547 \\
 & \textsc{C2SaferRust} & 214 & 540 &  9{,}250 & 4{,}663 & 2{,}212 & 53.6\% & \textbf{100\%} & \underline{394} \\
 & \textsc{EvoC2Rust}   &  \textbf{91} & \underline{437} & \textbf{6{,}022} & \underline{3{,}107} & \textbf{1{,}383} & \underline{58\%} & 0\% & N/A \\
 & \textsc{Encrust}     & \underline{135} & \textbf{295} & \underline{6{,}645} & \textbf{3{,}090} & \underline{2{,}099} & \textbf{96.2\%} & \textbf{100\%} & \textbf{211} \\
\cmidrule(l){1-10}
\multirow{4}{*}{tail}
 & \textsc{C2Rust}      & 389 & 1{,}092 & 11{,}663 & 5{,}869 & 2{,}580 & N/A    & \textbf{100\%} & 718 \\
 & \textsc{C2SaferRust} & 297 &   847 &  8{,}818 & \underline{3{,}739} & 2{,}433 & \underline{61.7\%} & \textbf{100\%} & \underline{469} \\
 & \textsc{EvoC2Rust}   & \textbf{173} & \underline{812} & \underline{7{,}870} & 3{,}805 & \textbf{1{,}918} & 49\% & 0\% & N/A \\
 & \textsc{Encrust}     & \underline{212} & \textbf{559} & \textbf{6{,}352} & \textbf{2{,}383} & \underline{2{,}231} & \textbf{96.8\%} & \textbf{100\%} & \textbf{266} \\
\cmidrule(l){1-10}
\multirow{4}{*}{truncate}
 & \textsc{C2Rust}      & 156 & 326 & 5{,}544 & 3{,}357 & 1{,}040 & N/A    & \textbf{100\%} & 386 \\
 & \textsc{C2SaferRust} & 124 & 263 & 4{,}443 & 2{,}521 &   \underline{952}   & \underline{70.2\%} & \textbf{100\%} & \underline{241} \\
 & \textsc{EvoC2Rust}   &  \textbf{49} & \underline{203} & \textbf{2{,}808} & \underline{1{,}733} & \textbf{543} & 66\% & 0\% & N/A \\
 & \textsc{Encrust}     &  \underline{76} & \textbf{174} & \underline{3{,}427} & \textbf{1{,}586} & 1{,}074 & \textbf{94.0\%} & \textbf{100\%} & \textbf{214} \\
\cmidrule(l){1-10}
\multirow{4}{*}{uniq}
 & \textsc{C2Rust}      & 227 & 343 & 6{,}066 & 3{,}590 & 1{,}150 & N/A    & \textbf{100\%} & 470 \\
 & \textsc{C2SaferRust} & 166 & 250 & 4{,}397 & 2{,}340 & \underline{1{,}025} & \underline{68.9\%} & \textbf{100\%} & \textbf{293} \\
 & \textsc{EvoC2Rust}   &  \textbf{67} & \underline{215} & \underline{2{,}656} & \underline{1{,}429} & \textbf{556} & 64\% & 0\% & N/A \\
 & \textsc{Encrust}     & \underline{122} & \textbf{87} & \textbf{2{,}469} & \textbf{1{,}070} & 1{,}037 & \textbf{95.9\%} & \textbf{100\%} & \underline{349} \\
\midrule
\multirow{4}{*}{\textit{Total}}
 & \textsc{C2Rust}      & 1{,}572 & 3{,}471 & 50{,}668 & 27{,}647 & 10{,}414 & N/A    & \textbf{100\%} & 2{,}951 \\
 & \textsc{C2SaferRust} & 1{,}191 & 2{,}716 & 38{,}947 & 19{,}282 &  9{,}594 & \underline{66.0\%} & \textbf{100\%} & \underline{2{,}087} \\
 & \textsc{EvoC2Rust}   &   \textbf{588} & \underline{2{,}385} & \underline{26{,}975} & \underline{13{,}530} & \textbf{6{,}375} & 59.7\% & 0\% & N/A \\
 & \textsc{Encrust}     &   \underline{829} & \textbf{1{,}489} & \textbf{25{,}724} & \textbf{10{,}626} & \underline{9{,}336} & \textbf{96.4\%} & \textbf{100\%} & \textbf{1{,}495} \\
\bottomrule
\end{tabular}}
\end{table*}

% ---------------------------------------------------------------
\subsection{Evaluation}
\label{sec:safety}
% ---------------------------------------------------------------

Tables~\ref{tab:safety-coreutils} and~\ref{tab:safety-laertes} report
the five safety metrics for all four systems on the Coreutils and Laertes
benchmarks respectively; Table~\ref{tab:safety-laertes} additionally reports
the function compilance pass rate, functional correctness, and Clippy warning count.

\paragraph{Safety reduction on Coreutils.}
Across all seven Coreutils programs \textsc{Encrust} reduces raw pointer
dereferences by 55\% and unsafe type casts by 60\% relative to the
\textsc{C2Rust} baseline, the two metrics where the wrapper pattern is most
effective: safe inner functions eliminate pointer arithmetic entirely, and
struct materialisation via \texttt{From} removes explicit \texttt{as} casts.
Raw pointer declarations and unsafe lines of code are reduced by 44\% and
46\% respectively.
Unsafe call expressions show only a 7\% reduction, because each wrapper
function itself constitutes an unsafe call site; this metric therefore
improves primarily through Phase~2 wrapper removal.

Compared with \textsc{C2SaferRust}, \textsc{Encrust} reduces raw pointer
dereferences by a further 42\% and unsafe type casts by a further 42\%,
reflecting the structural advantage of the wrapper pattern over
\textsc{C2SaferRust}'s inline pointer annotation approach.
\textsc{EvoC2Rust} reports nominally lower raw pointer declaration, unsafe
lines, and unsafe call expression totals, but produces output that does not
compile (functional correctness = 0\% on all seven programs); its safety
numbers are therefore not directly comparable.
\textsc{Encrust} is the only system achieving both unsafe-construct
reduction and 100\% test correctness across the full Coreutils benchmark.

\paragraph{Safety reduction on Laertes.}
Results for Laertes cover all eight libraries.
Over the eight completed libraries, \textsc{Encrust} reduces raw pointer
dereferences by 57\% and unsafe type casts by 38\% relative to \textsc{C2Rust},
consistent with the Coreutils trend.
\textit{qsort} and \textit{snudown} show the largest relative reductions
(raw pointer dereferences near zero for qsort), while \textit{lil} and
\textit{urlparser} show more modest gains owing to their heavy use of
multi-level pointers and function callbacks that the wrapper pattern does
not fully eliminate.
\textsc{EvoC2Rust} achieves lower totals on the Laertes suite, but
correctness verification is unavailable for its output on these libraries,
so the comparison is one-dimensional.

\begin{table*}[htbp]
\centering
\caption{Safety and correctness metrics on the \textbf{Laertes} benchmark.
  Column definitions as in Table~\ref{tab:safety-coreutils}.
  \emph{Comp.Rate}: fraction of LLM-translated functions that compiled within
  the retry budget (\textsc{Encrust} and \textsc{EvoC2Rust});
  N/A for \textsc{C2Rust} (no LLM translation).
  \emph{Correct.}: test-vector pass rate (100\% = all inputs match C reference output).
  All 8 libraries have \textsc{Encrust} results.
  Best value per metric per library \textbf{bolded}; second-best \underline{underlined}.}
\label{tab:safety-laertes}
\resizebox{\textwidth}{!}{%
\begin{tabular}{ll rrrrr r c r}
\toprule
Library & System & Ptr.Decl $\downarrow$ & Ptr.Deref $\downarrow$ & Unsafe LoC $\downarrow$ & Unsafe Cast $\downarrow$ & Unsafe Call $\downarrow$ & Comp.Rate $\uparrow$ & Correct. $\uparrow$ & Clippy $\downarrow$ \\
\midrule
\multirow{4}{*}{bzip2}
 & \textsc{C2Rust}      & 147 & 3{,}614 & 9{,}304 & 6{,}601 & 1{,}686 & N/A  & \textbf{100\%} & \underline{254} \\
 & \textsc{C2SaferRust} &  85 & 2{,}217 & 7{,}688 & 3{,}397 & 1{,}252 & 57.58  & \textbf{100\%} & 417 \\
 & \textsc{EvoC2Rust}   & \textbf{65} & \underline{391} & \textbf{1{,}459} & \textbf{548} & \textbf{327} & \underline{83.00}  &   0\% & N/A \\
 & \textsc{Encrust}     &  \underline{79} & \textbf{384} & \underline{1{,}907} & \underline{886} & \underline{572} & \textbf{93.94} & \textbf{100\%} & \textbf{146} \\
\cmidrule(l){1-10}
\multirow{4}{*}{genann}
 & \textsc{C2Rust}      &  48 & 158 & 557 & 231 & 116 & N/A  & \textbf{100\%} & \underline{32} \\
 & \textsc{C2SaferRust} &  38 & 127 & 479 & 174 & \underline{102} & 57.14  & \textbf{100\%} & \textbf{29} \\
 & \textsc{EvoC2Rust}   &  \textbf{3} &   \textbf{6} & \textbf{111} &  \textbf{57} &  \textbf{19} & \underline{80.65}  &   0\% & N/A \\
 & \textsc{Encrust}     & \underline{19} & \underline{39} & \underline{333} & \underline{92} & 127 & \textbf{92.86} & \textbf{100\%} & 59 \\
\cmidrule(l){1-10}
\multirow{4}{*}{lil}
 & \textsc{C2Rust}      & 416 & 1{,}674 & 5{,}416 & 2{,}231 & 1{,}729 & N/A  & \textbf{100\%} & \textbf{583} \\
 & \textsc{C2SaferRust} & 391 &   983 & 4{,}519 & 1{,}241 & \underline{1{,}651}   & \underline{47.65}  & \textbf{100\%} & \underline{670} \\
 & \textsc{EvoC2Rust}   &   \textbf{4} &   \textbf{2} &  \textbf{61} &  \textbf{25} &  \textbf{23} & 23.28  &   0\% & N/A \\
 & \textsc{Encrust}     & \underline{382} & \underline{887} & \underline{3{,}012} & \underline{800} & 1{,}751 & \textbf{97.32} & \textbf{100\%} & 698 \\
\cmidrule(l){1-10}
\multirow{4}{*}{optipng}
 & \textsc{C2Rust}      & 1{,}157 & 5{,}297 & 49{,}444 & 26{,}824 & 7{,}044 & N/A  & \textbf{100\%} & \underline{4{,}964} \\
 & \textsc{C2SaferRust} & \textbf{625} & \underline{3{,}452} & \textbf{34{,}178} & \underline{20{,}744} & \underline{5{,}974} & \underline{65.32}  & \textbf{100\%} & 4{,}984 \\
 & \textsc{EvoC2Rust}   &   820 & 4{,}488 & 40{,}486 & 22{,}775 & \textbf{5{,}217} & 63.59  &   0\% & N/A \\
 & \textsc{Encrust}     & \underline{634} & \textbf{2{,}087} & \underline{37{,}111} & \textbf{18{,}548} & 6{,}117 & \textbf{96.76} & \textbf{100\%} & \textbf{4{,}555} \\
\cmidrule(l){1-10}
\multirow{4}{*}{qsort}
 & \textsc{C2Rust}      &   4 &  11 &  67 &  34 &  21 & N/A  & \textbf{100\%} & \underline{7} \\
 & \textsc{C2SaferRust} &   \underline{1} &   \textbf{0} &  \textbf{14} &   \textbf{5} &   \textbf{4} & 80.00  & \textbf{100\%} & \textbf{6} \\
 & \textsc{EvoC2Rust}   &   \textbf{0} &   \underline{1} &  \underline{26} &  21 &   \underline{8} & \underline{87.50}  &   0\% & N/A \\
 & \textsc{Encrust}     &   \underline{1} &   \textbf{0} &  45 &  \underline{19} &  22 & \textbf{100.00} & \textbf{100\%} & 16 \\
\cmidrule(l){1-10}
\multirow{4}{*}{snudown}
 & \textsc{C2Rust}      &  56 & 287 & 1{,}678 & 1{,}809 & 449 & N/A  & \textbf{100\%} & 238 \\
 & \textsc{C2SaferRust} &  \underline{21} & 113 &   869 &   591 & 313   & 74.19  & \textbf{100\%} & \underline{96} \\
 & \textsc{EvoC2Rust}   &   \textbf{2} &  \textbf{11} &  \textbf{42} &  \textbf{30} &  \textbf{17} & \underline{82.66}  &   0\% & N/A \\
 & \textsc{Encrust}     &  25 &  \underline{99} &   \underline{410} &   \underline{209} & \underline{199}   & \textbf{93.55} & \textbf{100\%} & \textbf{69} \\
\cmidrule(l){1-10}
\multirow{4}{*}{tulipindicators}
 & \textsc{C2Rust}      & 1{,}081 & 3{,}047 & 19{,}408 & 12{,}349 & 3{,}336 & N/A  & \textbf{100\%} & 943 \\
 & \textsc{C2SaferRust} &   977 & 2{,}652 & 18{,}426 & 11{,}499 & 3{,}045   & 42.43  & \textbf{100\%} & \underline{727} \\
 & \textsc{EvoC2Rust}   & \textbf{149} & \textbf{119} & \textbf{702} & \textbf{291} & \textbf{157} & \underline{53.49}  &   0\% & N/A \\
 & \textsc{Encrust}     &   \underline{865} & \underline{1{,}169} & \underline{15{,}781} & \underline{8{,}716} & \underline{3{,}031}   & \textbf{92.25} & \textbf{100\%} & \textbf{714} \\
\cmidrule(l){1-10}
\multirow{4}{*}{urlparser}
 & \textsc{C2Rust}      &  82 &  89 & 1{,}627 &   \underline{573} & 832 & N/A  & \textbf{100\%} & \underline{153} \\
 & \textsc{C2SaferRust} &  \underline{76} &  \textbf{36} & \textbf{1{,}039} &   \textbf{394} & \underline{517} & 36.36  & \textbf{100\%} & \textbf{105} \\
 & \textsc{EvoC2Rust}   &   \textbf{4} &  \textbf{36} & \underline{1{,}090} &   575 & \textbf{295} & \underline{90.00}  &   0\% & N/A \\
 & \textsc{Encrust}     &  \underline{76} &  \underline{64} & 1{,}591 &   771 & 663  & \textbf{90.91} & \textbf{100\%} & 196 \\
\midrule
\multirow{4}{*}{\textit{Total}}
 & \textsc{C2Rust}      & 2{,}991 & 14{,}177 & 87{,}501 & 50{,}652 & 15{,}213 & N/A  & \textbf{100\%} & 7{,}174 \\
 & \textsc{C2SaferRust} & 2{,}214 &  9{,}580 & 67{,}212 & 38{,}045 & 12{,}858  & 56.07  & \textbf{100\%} & \underline{7{,}034} \\
 & \textsc{EvoC2Rust}   & \textbf{1{,}047} & \underline{5{,}054} & \textbf{43{,}977} & \textbf{24{,}322} & \textbf{6{,}063} & \underline{61.74}  &   0\% & N/A \\
 & \textsc{Encrust}     & \underline{2{,}081} &  \textbf{4{,}729} & \underline{60{,}190} & \underline{30{,}041} & \underline{12{,}482}  & \textbf{95.09} & \textbf{100\%} & \textbf{6{,}453} \\
\bottomrule
\end{tabular}}
\end{table*}

\paragraph{Phase~1 translation coverage.}
Across all seven Coreutils programs, \textsc{Encrust} successfully translated
999 of 1,097 functions in Phase~1 (within the five-retry budget), and
Phase~2 \texttt{function\_translate} tasks recovered an additional 59
functions, yielding an overall function compilance pass rate of \textbf{96.4\%} (1,058/1,097).
Per-program rates range from 94.0\% (\texttt{truncate}) to 99.2\%
(\texttt{pwd}), with the remaining 3.6\% of functions retained as their
original unsafe bodies because every LLM attempt either failed to compile
or failed the test-vector suite within the combined retry budget of both
phases.

\paragraph{Functional correctness.}
\textsc{Encrust} achieves 100\% correctness on all seven Coreutils programs
and all eight Laertes libraries (Tables~\ref{tab:safety-coreutils}--\ref{tab:safety-laertes}),
confirming that the Live Scaffold Invariant (\S\ref{sec:compile-test}) prevents
any behavioral regression throughout translation.
By contrast, \textsc{EvoC2Rust} scores 0\% on both benchmarks because its
translated projects do not compile.

\begin{table}[htbp]
\centering
\caption{Ablation: incremental contribution of each phase, aggregated
  over the Coreutils (7 programs) and Laertes (8 libraries) benchmarks.
  \textsc{Encrust} Phase~1 = after Type-Directed Wrapper Elimination only;
  \textsc{Encrust} (Phase~1+2) = full pipeline.
  All metrics $\downarrow$ lower is better.}
\label{tab:ablation}
\small
\setlength{\tabcolsep}{4pt}
\begin{tabular}{ll rrrrr r}
\toprule
Suite & Configuration & Ptr.Decl & Ptr.Deref & Unsafe LoC & Unsafe Cast & Unsafe Call & Clippy \\
\midrule
\multirow{3}{*}{Coreutils}
 & \textsc{C2Rust}                      & 1{,}572 & 3{,}471 & 50{,}668 & 27{,}647 & 10{,}414 & 2{,}951 \\
 & \textsc{Encrust} Phase~1             &   886 & 1{,}585 & 27{,}424 & 11{,}227 &  9{,}659 & 1{,}463 \\
 & \textsc{Encrust} (Phase~1+2)         &   829 & 1{,}489 & 25{,}724 & 10{,}626 &  9{,}336 & 1{,}495 \\
\cmidrule(l){1-8}
\multirow{3}{*}{Laertes}
 & \textsc{C2Rust}                      & 2{,}991 & 14{,}177 & 87{,}501 & 50{,}652 & 15{,}213 & 7{,}174 \\
 & \textsc{Encrust} Phase~1             & 2{,}269 &  8{,}278 & 69{,}009 & 36{,}107 & 13{,}795 & 6{,}465 \\
 & \textsc{Encrust} (Phase~1+2)         & 2{,}081 &  4{,}729 & 60{,}190 & 30{,}041 & 12{,}482 & 6{,}453 \\
\bottomrule
\end{tabular}
\end{table}

% ---------------------------------------------------------------

\begin{table}[htbp]
\centering
\caption{Phase~2 task completion breakdown aggregated over all 15 programs
  (7 Coreutils + 8 Laertes).
  \emph{Rate} = Completed / Discovered.
  \emph{Avg.\ iters} = mean agentic loop iterations per completed task.}
\label{tab:phase2-tasks}
\small
\begin{tabular}{lrrrr r}
\toprule
Task type & Disc. & Compl. & Failed & Rate & Avg.\ iters \\
\midrule
\texttt{static\_mut}          & 363 & 227 & 136 & 62.5\% & 21.8 \\
\texttt{wrapper\_removal}     & 454 & 347 & 107 & 76.4\% & 20.6 \\
\texttt{struct\_migration}    &  23 &  21 &   2 & 91.3\% & 10.7 \\
\texttt{struct\_use\_migrate} &  15 &  14 &   1 & 93.3\% &  3.8 \\
\texttt{function\_translate}  & 263 & 174 &  89 & 66.2\% & 23.7 \\
\texttt{dead\_code}           &  15 &   7 &   8 & 46.7\% & 18.3 \\
\midrule
\textbf{Total}                & 1{,}133 & 790 & 343 & 69.7\% & 21.0 \\
\bottomrule
\end{tabular}
\end{table}
\subsection{Ablation Study: Phase~2 Contribution}
\label{sec:phase2-contribution}
% ---------------------------------------------------------------

To isolate the contribution of each phase, Table~\ref{tab:ablation}
compares three configurations on the five safety metrics and Clippy
warning count: \textsc{C2Rust} (no translation), Phase~1 only
(\texttt{rust\_safe\_remap/}, after Type-Directed Wrapper Elimination),
and the full \textsc{Encrust} pipeline (Phase~1 + Phase~2).
Phase~1 alone reduces unsafe lines of code by 45.9\% on Coreutils and
21.1\% on Laertes relative to \textsc{C2Rust}, with the TDWE remapping
step accepting 79.9\% (1{,}210/1{,}514) of wrapper--safe function pairs.
Raw pointer dereferences see the largest Phase~1 reduction (54.3\% on
Coreutils, 41.6\% on Laertes), while unsafe call expressions improve only
modestly (7.2\% and 9.3\% respectively), since each retained wrapper
constitutes an unsafe call site.
Phase~2 adds a further 6.1\% relative reduction in raw pointer
dereferences on Coreutils (total 57.1\% over \textsc{C2Rust}) and 42.9\%
on Laertes (total 66.6\%), and reduces unsafe lines of code by a further
6.2\% and 12.8\% respectively, with all five metrics improving across
both benchmark suites.

Table~\ref{tab:phase2-tasks} shows the Phase~2 task breakdown aggregated
over all 15 programs.
\textsc{Encrust} discovers 1{,}133 tasks and completes 790 (69.7\%),
averaging 21.0 agentic loop iterations per completed task.
Structural tasks achieve the highest rates: \texttt{struct\_migration}
at 91.3\% and \texttt{struct\_use\_migrate} at 93.3\%, owing to their
localised, well-typed rewrites (mean 10.7 and 3.8 iterations).
\texttt{dead\_code} cleanup is lowest (46.7\%) because whole-project
reachability reasoning frequently exceeds the 40-iteration budget.
\texttt{wrapper\_removal} is the most frequent task type (454 tasks)
and achieves 76.4\%, as most residual wrapper--safe pairs have tractable
type mismatches given full codebase context.

\section{Limitations}
\label{sec:limitations}
\emph{Test-vector coverage.}
Correctness is verified only for code paths exercised by the test suite;
functions never called during testing are translated without behavioral
verification.
As reported in Tables~\ref{tab:coreutils-stats} and~\ref{tab:laertes-stats},
the covered fraction varies substantially across programs, so the true
correctness of the full translated codebase cannot be guaranteed beyond
the tested paths.

\emph{TDWE completeness and agentic scope.}
Phase~1's Type-Directed Wrapper Elimination operates on a best-effort basis:
it successfully eliminates wrapper--safe function pairs for 79.9\% of
translated functions, but does not guarantee that the resulting safe
function bodies are themselves free of residual unsafe constructs such as
raw pointer manipulation or unchecked indexing.
Crucially, functions whose wrappers are removed by TDWE are \emph{not}
subsequently re-examined by the Phase~2 agentic loop, which targets only
the task categories discovered in the post-TDWE codebase
(\texttt{static\_mut}, \texttt{wrapper\_removal}, \texttt{struct\_migration},
\texttt{function\_translate}, \texttt{dead\_code}).
Any unsafe code that survives inside TDWE-translated function bodies outside
these categories therefore remains in the final output without further
refinement.

\section{Conclusion}
\label{sec:conclusion}

We presented \textsc{ENCRUST}, a two-phase pipeline that translates real-world C
programs to safe Rust while guaranteeing behavioral equivalence throughout.
The first phase, \emph{Encapsulated Substitution}, introduces an ABI-preserving
wrapper pattern that decouples per-function type-signature changes from their
call sites, enabling independent LLM-driven translation with silent rollback,
followed by a deterministic type-directed wrapper elimination pass.
The second phase, \emph{Agentic Refinement}, targets the residual unsafe
constructs that exceed per-function scope, namely \texttt{static mut} globals,
skipped wrapper pairs, and failed translations, through a tool-equipped LLM
agent operating on the whole codebase under a baseline-aware verification gate
that prevents correctness regressions.

Evaluated on 15 programs spanning 197,706 lines of C, \textsc{Encrust}
reduces unsafe lines of code by 37.8\% over the \textsc{C2Rust}
baseline while achieving 100\% test-vector correctness on all programs.
Phase~1 alone accounts for 34.3\% of this reduction; Phase~2 contributes
an additional 5.4\% relative reduction by targeting unsafe constructs that
are intractable within per-function scope by design, completing 790 of
1{,}133 discovered tasks at an average of 21.0 agentic loop iterations per task.
Across both benchmark suites, \textsc{Encrust} surpasses \textsc{C2SaferRust}
on unsafe-construct reduction while maintaining the same 100\% test correctness,
and unlike \textsc{EvoC2Rust}, achieves this reduction without sacrificing
functional equivalence.
These results demonstrate that decomposing the translation problem into
encapsulated per-function substitution followed by whole-codebase agentic
refinement is an effective strategy for scaling correctness-preserving
C-to-safe-Rust translation to real-world programs.

Several directions remain open.
Extending Phase~1 to handle inline assembly and \texttt{setjmp}/\texttt{longjmp}
would widen the class of C programs within scope.
Replacing GPT-4o with open-weight models would reduce cost and improve
reproducibility across model updates.
Finally, integrating a lightweight static pointer-provenance analysis could
expand the set of wrapper pairs eligible for elimination, closing the gap
between Phase~1-only and full \textsc{Encrust} output on programs with
heavy pointer arithmetic.

%\begin{acks}
%To Robert, for the bagels and explaining CMYK and color spaces.
%\end{acks}

\section*{Data-Availability Statement}
The test cases and data for Coreutils used in this paper are taken 
from the publicly available C2SaferRust repository~\cite{c2saferrust}. 
The test vectors for Laertes will be made available on GitHub upon publication.

\bibliographystyle{ACM-Reference-Format}
\bibliography{acmart}

\end{document}